\documentclass[12pt, draftclsnofoot, onecolumn]{IEEEtran}
\usepackage{bm,cite,algorithm,float,amsmath,amssymb}

\usepackage{graphicx, epstopdf}
\usepackage{graphicx}
\usepackage{cite}
\usepackage{citesort}
\usepackage{balance}
\usepackage[utf8x]{inputenc}
\usepackage{fancyhdr}
\usepackage{subfigure}
\usepackage{dsfont}
\usepackage{calc}

\usepackage[compatible]{algpseudocode}

\usepackage{tikz}
\usetikzlibrary{shapes.geometric, arrows}
\usepackage{flowchart}

\usepackage{caption}
\usepackage{amsmath,amssymb,amsthm}
\usepackage{graphicx,epstopdf}
\usepackage{epsfig}	
\usepackage{amsfonts,balance}

\usepackage{lipsum}

\usepackage{xcolor,colortbl}
\definecolor{Gray}{gray}{0.9}

\newtheorem{lemma}{Lemma}

\IEEEoverridecommandlockouts

\usepackage{graphicx,epstopdf}
\usepackage{epsfig}	
\usepackage{amsfonts,balance}
\usepackage{bbm}
\floatname{algorithm}{Algorithm}
\usepackage{lipsum}
\usepackage{url}
\newtheorem{theorem}{Theorem}

\newtheorem{corollary}{Corollary}

\begin{document}

\title{Random Access Analysis for Massive IoT Networks under A New Spatio-Temporal Model: A Stochastic Geometry Approach}

%Spatio-Temporal Model and Random Access Analysis for Massive IoT Networks: A Stochastic Geometry Approach

\author{Nan Jiang, Yansha Deng, Xin Kang, and Arumugam Nallanathan
\\
%{\small ${}^*$Department of Informatics, King's College London, London, UK\\
%${}^{\ddagger}$ University of Electronic Science and Technology of China, Chengdu, Sichuan, China\\}

\thanks{N. Jiang, Y. Deng and A. Nallanathan are with Department of Informatics, King's College London, London, UK 
(e-mail:\{nan.3.jiang, yansha.deng, arumugam.nallanathan\}@kcl.ac.uk). (Corresponding author: Yansha Deng (e-mail:yansha.deng@kcl.ac.uk).)}
\thanks{X. Kang is with National Key Laboratory of Science and Technology on communications,  University of Electronic Science and Technology of China, Chengdu, 611731, Sichuan, China (e-mail: kangxin83@gmail.com).}

}

\maketitle

\vspace*{-2.0cm}
\begin{abstract}
\vspace*{-0.2cm}
Massive Internet of Things (mIoT) has provided an auspicious opportunity to build powerful and ubiquitous connections that faces a plethora of new challenges, where cellular networks are potential solutions due to their high scalability, reliability, and efficiency. The Random Access CHannel (RACH) procedure is the first step of connection establishment between IoT devices and Base Stations (BSs) in the cellular-based mIoT network, where modeling the interactions between static properties of physical layer network and dynamic properties of queue evolving in each IoT device are challenging. To tackle this, we provide a novel traffic-aware spatio-temporal model to analyze RACH in cellular-based mIoT networks, where the physical layer network is modeled and analyzed based on stochastic geometry in the spatial domain, and the queue evolution is analyzed based on probability theory in the time domain. For performance evaluation, we derive the exact expressions for the preamble transmission success probabilities of a randomly chosen IoT device with different RACH schemes in each time slot, which offer insights into effectiveness of each RACH scheme. Our derived analytical results are verified by the realistic simulations capturing the evolution of packets in each IoT device. This mathematical model and analytical framework can be applied to evaluate the performance of other types of RACH schemes in the cellular-based networks by simply integrating its preamble transmission principle.

%Specifically, we derive the closed-form expressions for the preamble detection probability of a randomly chosen BS, the preamble transmission success probability of a randomly chosen IoT device, and the transmission capacity per BS per preamble, as well as evaluate the performance of different RACH transmission schemes varying over time and offer insights into effectiveness of each transmission scheme. Note that, different from previous works, we perform the practical simulations capturing the evolution of packets in each IoT device. This model can also be applied to evaluate the performance of other types of RACH transmission schemes in the cellular-based m-IoT networks by simply modifying its packets evolution principle, and other types of cellular-based networks for analyzing the structure of spatio-temporal pattern.

\end{abstract}

\vspace*{-0.5cm}
%========================================================================
\section{INTRODUCTION}
\vspace*{-0.2cm}

%Massive Internet of Things (m-IoT) is deemed to create a revolution in our future network, which is expected that huge numbers of devices accessing into wireless network and exchange data. It is expected that IoT devices will expand to approximate 26 billion in 2020 \cite{nokia2016white,al2015internet}, and a great number of new applications will emerge, such as autonomous driving, remote health care, smart-homes, smart-grids, and etc, which can be used to improve industries or society by enhancing processes and enabling automation. For the purpose to enable such applications, data should be exchanged among objects, machines and human smartly and automatically.

Massive Internet of Things (mIoT) is deemed to connect billions of miscellaneous mobile devices or IoT devices that empowers individuals and industries to achieve their full potential. A plethora of new applications, such as autonomous driving, remote health care, smart-homes, smart-grids, and etc, are being innovated via mIoT, in which ubiquitous connectivities among massive IoT devices are operated fully automatedly without human intervention. The successful operation of these IoT applications faces various challenges, among them providing wireless access for the tremendous number of IoT devices has been considered to be the main problem. This issue has been regarded as one of key differences between mIoT and human-to-human (H2H) wireless communication networks, such that the conventional H2H communication architecture needs to be adjusted to support the mIoT networks.

%The successful operation of these IoT applications have unique requirement in those tangible performance, such as data rate, latency, energy consumption, and hardware cost, and also in those intangible performance, such as quality of services (QoS), supporting massive spatially-spread uplink transmission. Innovative new technologies need to be urgently developed to provide smart and efficient m-IoT networks.

%Typically, the traditional cellular networks do not need to associate with a massive number of simultaneously access requests, which may always occur in the IoT application above the capable that the current cellular network can deal with. The Medium Access Control (MAC) layer is the lower sublayer of the layer 2 (data link layer), which is used to addressing channel access control. Many previous works about to enable the massive IoT access requests has focused on the design and analysis of congestion control scheme of RACH in MAC layer \cite{alavikia2016multiple,lien2012cooperative,alavikia2016multiple,leyva2016performance,lin2016estimation,polese2016evaluation}. However, The modeling of the physical layer of RACH, in terms of analyzing the probability of access request being detected considering the effects of radio channels such as path-loss, power control, and fading, was few developed. In the m-IoT networks, massive simultaneous access requests may generate large interference which leads to a low detection probability of access request, and it is worth to develop novel frameworks to analyze the RACH considering both MAC and physical layer.

Previously, cellular network (e.g, Long Term Evolution (LTE)) and short-range transmission technologies (e.g, ZigBee, Bluetooth) were considered as potential solutions to support mIoT networks, however none of them can achieve all wide coverage, low power consumption and supporting massive IoT devices at the same time \cite{zanella2014internet,gubbi2013internet,hasan2013random,laya2014random}. 
To solve this, Low-Power Wide Area Networks (LPWANs) is proposed as an alternative solution for mIoT networks that enables the operation in the unlicensed band (e.g, LoRa, Sigfox) and licensed band (e.g, extended coverage GSM-IoT, enhanced machine type communication, and narrow band IoT (NB-IoT)). According to the Third Generation Partnership Project (3GPP), the IoT technologies are suggested to be developed based on the existing cellular infrastructure, due to its low additional hardware deployment cost as well as high-level of security by operating on the licensed band \cite{3GPP2011Study,hasan2013random,laya2014random,Gotsis2012M2M,Ksentini2012Cellular,nokia2016white,ericsson2016white}. %It is reported that cellular-based mIoT networks have been commercially tested in 2016 by China Mobile, Huawei, ZTE, and Ericsson, and also been initially deployed in Yintan, China in early 2017.

In the cellular-based mIoT network, connections between IoT device with BS are provided by incorporating these IoT devices in existing cellular networks directly or via IoT gateways. In this network, the number of IoT devices
is expected to raise up to more than thirty thousands per cell and such IoT devices may request access simultaneously for their small size data packets uplink transmission \cite{3GPP1046632010RACH,3GPP1037422010RACH,3GPP2011Study}. As such improving the access mechanisms of current cellular systems is one of key challenges for the cellular-based mIoT network \cite{3GPP2011Study,Ksentini2012Cellular,Gotsis2012M2M,biral2015challenges,laya2014random,hasan2013random}. In LTE, a device performs Random Access CHannel (RACH) procedure when it needs to establish or re-establish a data connection with its associated BS, and the first step of RACH is that the device transmits a preamble via physical random access channel (PRACH) \cite{LTE2013dahlman}. Two ways exist for accessing to the network: 1) the contention-free RACH for delayed-constrained access requests (e.g, handover), where the BS distributes one of the reserved dedicated preamble to a device, and then the device uses its dedicated preamble to initiate a contention-free RACH; 2) the contention-based RACH for delay-tolerant access requests (e.g, data transmission), where an IoT device randomly chooses a preamble from non-dedicated preambles to transmit to its associated BS \cite{LTE2013dahlman}. Generally, the contention-based RACH is much more sensitive to IoT traffic \cite{biral2015challenges,hasan2013random,laya2014random}, such that most works have analyzed its scalability characteristics in supporting massive concurrent access requests \cite{alavikia2016multiple,gotsis2012evolution,leyva2016performance,lin2016estimation,zheng2012radio,dama2016novel,wali2017optimal}.

The contention-based RACH has been widely studied in the conventional LTE networks, where the most critical point of this issue concerns modeling and analyzing time-varying queues and RACH schemes in MAC layer \cite{zhou2008efficient,wei2015modeling}. Recently, a number of studies have been launched to discuss whether the contention-based RACH of LTE is suitable for mIoT, and how to evolve cellular systems to provide efficient access for mIoT networks \cite{laya2014random,lin2016estimation,dama2016novel,wali2017optimal}. In \cite{lin2016estimation}, the authors developed a MAC-level model for the four-step RACH procedure to analyze and compare the baseline scheme and the dynamic back-off scheme. In \cite{dama2016novel}, a novel Access Class Barring (ACB) scheme is proposed with a constant association rate. In \cite{wali2017optimal}, the authors devolop analytical models for the spatial-randomization ACB scheme and the time-spatial randomization, as well as compare them with the 3GPP specified eNodeB-employed time-randomization ACB scheme. However, in \cite{zhou2008efficient,wei2015modeling,lin2016estimation,dama2016novel,wali2017optimal}, the collision events are considered as the main outage condition, and the preamble transmission failure impacted by the physical channel propagation characteristics is simplified. Generally speaking, in the large-scale cellular-based mIoT network, the physical layer characteristics can strongly influence the performance of RACH success, due to that the received signal-to-interference-plus-noise ratio (SINR) at the BS can be severely degraded by the mutual interference generated from massive IoT devices. In this scenario, the random positions of the transmitters make accurate modeling and analysis of this interference even more complicated.

Stochastic geometry has been regarded as a powerful tool to model and analyze mutual interference between transceivers in the wireless networks, such as conventional cellular networks \cite{novlan2013analytical,andrews2011tractable,S2013Downlink}, wireless sensor networks \cite{deng2016physical}, cognitive radio networks \cite{deng2016artificial,elsawy2013stochastic}, and heterogenous cellular networks \cite{deng2016modeling,HElSawy2014Stochastic,zhong2017heterogeneous}. However, there are two aspects that limit the application of conventional stochastic geometry analysis to the RACH analysis of the cellular-based mIoT networks: 1)  conventional stochastic geometry works focused on analyzing normal uplink and downlink data transmission channel, where the intra-cell interference is not considered, due to the ideal assumption that each orthogonal sub-channel is not reused in a cell, whereas massive IoT devices in a cell may randomly choose and transmit the same preamble using the same sub-channel; 2) these conventional stochastic geometry works only modeled the spatial distribution of transceivers, and ignored the interactions between static properties of physical layer network and the dynamic properties of queue evolving in each transmitter due to the assumptions of backlogged
network with saturated queues \cite{stamatiou2010random,zhong2016stability,gharbieh2016spatiotemporal}.

To model these aforementioned interactions, recent works have studied the stability of spatially spread interacting queues in the network based on stochastic geometry and queuing theory \cite{stamatiou2010random,zhong2016stability,zhong2017heterogeneous,gharbieh2016spatiotemporal}. The work in \cite{stamatiou2010random} is the first paper applying the stochastic geometry and queuing theory to analyze the performance of RACH in distributed networks, where each transmitter is composed of an infinite buffer, and its location is changed following a high mobility random walk. The work in \cite{zhong2016stability} investigated the stable packet arrival rate region of a discrete-time slotted RACH network, where the transceivers are static and distributed as independent Poisson point processes (PPPs). The work in \cite{zhong2017heterogeneous} analyzed the delay in the heterogeneous cellular networks with spatio-temporal random arrival of traffic, where the traffic of each device is modeled by a marked Poisson process, and the statistics of such traffic with different offloading policies are compared. In \cite{gharbieh2016spatiotemporal}, the authors have modeled the randomness in the locations of IoT devices and BSs via PPPs, and leveraged the discrete time Markov chain to model the queue and protocol states of each IoT device. However, the model is limited in capturing the dynamic preamble success probability during the time evolution, such that it can only derive the analytical result during the steady state, and this result is unable to be verified by simulations.

In this paper, we develop a novel spatio-temporal mathematical framework for cellular-based mIoT network using stochastic geometry and probability theory, where the BSs and IoT devices are modeled as independent PPPs in the spatial domain. In the time domain, the new arrival packets of each IoT device are modeled by independent Poisson arrival processes  \cite{zhong2017heterogeneous,chen2009analysis,gotsis2012evolution,zhou2013lte}. The packets status in each IoT device that are jointly populated by the new Poisson arrival packets and the accumulated packets in the previous time slots according to its stochastic geometry analysis, determines the aggregate interference at the received SINR in the current time slot, which then determines the non-empty probability and non-restrict probability of IoT device (i.e., IoT device have back-logged packets and permission to transmit currently) in the current time slot. The contributions of this paper can be summarized in the following points:

%The departure process depends on the spatial stochastic geometry model, that is influenced by the mutual interference among the active IoT devices (ie, IoT device have available packets and permission to transmit currently).

%Different from \cite{gharbieh2016spatiotemporal}, we will perform the practical simulation capturing the evolution of packets in each IoT device, and evaluate the traffic and network performance over time with simple analytical solutions.

 % we will present the network performances and the evolution of queues with time, and we will perform the network in variable traffic conditions (eg, bursty traffic arrival in a period), which would not be limited by the unstable network operation.

\begin{itemize}
  \item
We present a novel spatio-temporal mathematical framework for analyzing contention-based RACH of the mIoT network. Assuming the independent Poisson arrival, the packets accumulation and preamble transmission of a typical IoT device in each time slot is accurately modeled.

%, only rooted in the spatial model of devices.

%The paper presents a novel spatio-temporal mathematical model for analyzing RACH of m-IoT communication. This model equips a discrete time domain, which is capable of dealing with the spatial-temporal random arrival of traffic, queue delay, and the static property of networks of different time slots. From our model, the RACH of cellular-based m-IoT networks can be efficiently deployed, designed and analyzed, as well as this model can be developed and extended to analyze more complex RACH transmission schemes or other types of networks requiring spatio-temporal pattern.
  \item

With single time slot, we derive the exact expressions for the preamble detection probability of a randomly chosen BS, the preamble transmission success probability of a randomly chosen IoT device, and the number of received packets per BS in the cellular-based mIoT networks.

%We are able to derive a general expression for the probability of coverage in a cellular network where the interference fading/shadowing follows an arbitrary distribution.
%This paper analyzes and compares the preamble detection probabilities in the BS site and the preamble transmission success probabilities in the IoT site of PRACH, which is different with the normal uplink and downlink data transmission channel.
  \item
With multiple time slots, the queue statuses are firstly analyzed based on probability theory, and then approximated by their corresponding Poisson arrival distributions, which facilitates the queuing analysis. By doing so, we derive the exact expressions for the preamble transmission success probability of a randomly chosen IoT device in each time slot with the baseline, the ACB, and the back-off schemes for their performance comparison.
  \item
We develop a realistic simulation framework to capture the randomness location, preamble tranmission, and the real packets arrival, accumulation, and departure of each IoT device in each time slot, where the queue evolution as well as the stochastic geometry analysis are all verified by our proposed realistic simulation framework.
  \item
The analytical model presented in this paper can also be applied for the performance evaluation of other types of RACH schemes in the cellular-based networks by substituting its preamble transmission principle.

%From our model, the RACH of cellular-based m-IoT networks can be efficiently deployed, designed and analyzed, as well as this model can be developed and extended to analyze more complex RACH transmission schemes or other types of networks requiring spatio-temporal pattern.
\end{itemize}

The rest of the paper is organized as follows. Section II presents the network model. Sections III derives preamble detection probability of a randomly chosen BS and the preamble transmission success probability of a randomly chosen IoT device in single time slot. Section IV derives the preamble transmission success probabilities of a randomly chosen IoT device in each time slot with different schemes. Finally, Section V concludes the paper.

\vspace*{-0.3cm}
\section{SYSTEM MODEL}
\vspace*{-0.2cm}
We consider an uplink model for cellular-based mIoT network consists of a single class of base stations (BSs) and IoT devices, which are spatially distributed in  ${{\mathbb R}^2}$ following two independent homogeneous Poisson point process (PPP), ${\Phi_B }$ and $\Phi_D$, with intensities ${{\lambda_B}}$ and $\lambda_D$, respectively. Same as \cite{novlan2013analytical,HElSawy2014Stochastic,Sung2013Stochastic}, we assume each IoT device associates to its geographically closest BS, and thus forms a Voronoi tesselation, where the BSs are uniformly distributed in the Voronoi cell. 
%We provide a spatio-temporal traffic model for the cellular-based mIoT network: 1) the spatial distribution model of the IoT devices, and 2) the temporal model for the packets arrival at each IoT device. 
Same as \cite{zhong2016stability,zhong2017heterogeneous}, the time is slotted into discrete time slots, and the number and locations of BSs and IoT devices are fixed all time once they are deployed. 

% However, different from \cite{zhong2016heterogeneous} where the locations of the IoT devices are fixed all the time, we take into account the temporal variation of the locations of the IoT devices in each time slot due to their random access or mobility. To be more specific, the location of each IoT device choosing same preamble varies in each different time slot, due to the following two reasons: 1) the IoT devices may be with mobility; 2) each active IoT device randomly chooses a preamble in different time slot, such that the locations of active IoT devices choosing same preamble are changed in different time slot.
\vspace*{-0.5cm}
\subsection{Network Description}
\vspace*{-0.2cm}
We consider a standard power-law path-loss model, where the signal power decays at a rate ${r}^{ - \alpha }$ with the propagation distance $r$, and the path-loss exponent ${\alpha }$. We consider Rayleigh fading channel, where the channel power gains $h(x,y)$ between two generic locations $x,y \in {{\mathbb R}^2}$ is assumed to be exponentially distributed random variables with unit mean. All the channel gains are independent of each other, independent of the spatial locations, and identically distributed (i.i.d.). For the brevity of exposition, the spatial indices $(x,y)$ are dropped.

Uplink power control has been an essential technique in cellular network \cite{goldsmith2005wireless,HElSawy2014Stochastic,novlan2013analytical}. We assume that a full path-loss inversion power control is applied at all IoT devices, where each IoT device compensates for its own path-loss to keep the average received signal power equal to a same threshold $\rho$ \cite{3gpp2016Evolved,gharbieh2016spatiotemporal}. By doing so, as a user moves closer to the desired base station, the transmit power required to maintain the same received signal power decreases, which saves energy for battery-powered IoT devices. More importantly, it helps to solve the “near-far” problem, where a BS cannot decode the signals from cell-edge due to high aggregate interference from other nearby IoT devices. The transmit power of $i$th IoT device $P_i$ depends on the distance from its associated BS, and the defined threshold $\rho$, where $P_i = {\rho}{r_i}^{{ \alpha }}$. In order to successfully transmit a signal from the IoT device, the maximum transmit power should be high enough for its path-loss inversion, otherwise, it does not transmit the signal and goes into a truncation outage. Here, we assume that the density of BSs is high enough and none of the IoT device suffers from truncation outage (i.e., the transmit power of IoT device is large enough for uplink path-loss inversion, while not violating its own maximum transmit power constraint).

\vspace*{-0.3cm}
\subsection{Contention-Based Random Access Procedure}
\vspace*{-0.2cm}
In the cellular-based network, the first step to establish an air interface connection is delivering requests to the associated BS via RACH \cite{LTE2013dahlman}, where the contention-based RACH is favored by mIoT network for the initial association to the network, the transmission resources request, and the connection re-establishment during failure \cite{Gotsis2012M2M,biral2015challenges,laya2014random,hasan2013random}. The contention-based RACH has four steps: In step 1, each device randomly chooses a preamble (i.e., orthogonal pseudo code, such as Zadoff-Chu sequence\footnote{In LTE, there are 64 available preambles for RACH in each BS, which are generated on the side of IoT devices from 10 different root sequences \cite{LTE2013dahlman}. Generally, the preambles generated from the same root are completely orthogonal, and preambles generated from two different roots are nearly orthogonal \cite{LTE2013dahlman}. To mitigate the interference among preambles generated from different root sequences, some published literature have specifically studied such correlations, and some results shown that their proposed preamble detectors can asymptotically achieve almost interference-free detection performance (i.e., nearly orthogonal) \cite{kim2017enhanced}. However, this is beyond the scope of this paper, and to focus on studying spatio-temporal model for RACH, we assume different preambles are completely orthogonal.}) from avalaible preamble pool, and send to its associated BS via PRACH. In step 2, the IoT device sets a random access response (RAR) window and waits for the BS to response with an uplink grant in the RAR. In step 3, the IoT device that successfully receives its RAR transmits a radio resource control connection request with identity information to BS. In step 4, the BS transmits a RRC Connection Setup message to the IoT device. Note that, only within the step 1 preamble is transmitted via PRACH, but within other steps signals are transmitted via normal uplink and downlink data transmission channel. Further details on the RACH can be found in \cite{LTE2013dahlman}. 

In the step 1 of contention-based RACH, the IoT device randomly selects a preamble from a group of non-dedicated preambles defined by the BS. Without loss of generality, we assume that each BS has an available preamble pool with the same number of non-dedicated preambles $\xi$, known by its associated IoT devices. Each preamble has an equal probability ($1/\xi$) to be chosen by an IoT device, and the average density of the IoT devices using the same preamble is ${\lambda_{Dp} } =  {\lambda_D}/{\xi}$, where the ${\lambda_{Dp} }$ is measured with unit devices/preamble/$\rm {km^2}$.

In the cellular-based mIoT network, ${\lambda_{Dp} }$ is able to be a huge number, due to that the slotted-ALOHA system allows all IoT devices requesting for access in the first available opportunity. Once a huge number of IoT devices transmit preambles simultaneously, the network performance might degrade due to that the preambles cannot be detected or decoded by the BS \cite{hasan2013random,laya2014random}. Therefore, the contention of preamble in the step 1 becomes one of the main challenges in RACH \cite{leyva2016performance,lin2016estimation,gursu2017hybrid,soorki2017stochastic,gharbieh2016spatiotemporal,3GPP2011Study}. Same as \cite{lin2016estimation,gursu2017hybrid,soorki2017stochastic,gharbieh2016spatiotemporal}, we assume that the step 2, 3, and 4 of RA are always successful whenever the step 1 is successful. In other words, the RACH may fail due to the following two reasons: 1) a preamble cannot be recognized by the received BS, due to its low received SINR; 2) the BS successfully received two or more same preambles simultaneously, such that the collision occurs, and the BS cannot decode any collided preambles. %The 3GPP and organization members have investigated the preamble transmission failure problem of the RACH in mIoT network \cite{3GPP1037422010RACH,3GPP1046632010RACH,3GPP2011Study}.
In this work, we limit ourselves to single preamble transmission fail same as \cite{gharbieh2016spatiotemporal,zhong2016stability}, and leave the collision for our future work, thus we assume that a RACH procedure is always successful if the IoT device successfully transmits the preamble to its associated BS.
\vspace*{-0.3cm}
\subsection{Physical Random Access CHannel and Traffic Model}
\vspace*{-0.2cm}
We consider a time-slotted mIoT network, where the PRACH happens at the beginning of a time slot within a small time interval $\tau_c$, and the least time of a time slot (i.e., the time between any two PRACHs) is a gap interval duration $\tau_g$ for data transmission as shown in Fig. \ref{fig:1}. Generally, the PRACH is reserved in the uplink channel and repeated in the system with a certain period that specified by the BS. For instance, in the LTE network, the uplink resource reserved for PRACH has a bandwidth corresponding to six resource blocks (1.08 MHz), and the PRACH is repeated with a periodicity varies between every 1 to 20 ms \cite{laya2014random,LTE2013dahlman}. During the PRACH duration, each active IoT device will transmit a preamble to its associated BS to request uplink channel resources for packets transmission. Here, the active IoT device represents that an IoT device is with non-empty buffers (i.e., ${N_{New}^m}+{N_{Cum}^m}>0$, where ${N_{New}^m}$ is the number of new arrived packets, and ${N_{Cum}^m}$ is the number of accumulated packets) and without access restriction, which will be detailed in the following section.

Without loss of generality, we assume the size of buffer in each IoT device is infinite, and none of the packets will be dropped off. At the beginning of the PRACH in the $m$th time slot, each IoT device checks its buffer status to determine whether itself requires to attempt RACH as shown in the Fig. \ref{fig:1}. In detail, the buffer status (i.e., queuing packets) are determined by the new arrived packets, and the accumulated packets that unsuccessfully departs (i.e., unsuccessfully RACH attempts or never been scheduled) before the last time slot.

Once a RACH succeeds, the IoT device will transmit the corresponding data sequences with the scheduled uplink channel resources. Here, we interchangeably use packet to represent the data sequences. In each device, the packets are line a queue waiting to be transmitted, where each packet has the same priority, and the BSs are unaware of the queue status of their associated IoT devices. It is assumed that the BS will only schedule uplink channel resources for the head-of-line packet$^2$ and each IoT deivce transmits packets via a First Come First Serve (FCFS) packets scheduling scheme - the basic and the most simplest packet scheduling scheme, where all packets are treated equally by placing them at the end of the queue once they arrive \cite{gow2006mobile}.

\vspace*{-0.5cm}
\captionsetup{singlelinecheck=true}
\begin{table}[htbp!]
	\centering
	\caption{Packets Evolution in the Typical IoT Device.}
	{\renewcommand{\arraystretch}{0.6}
		\begin{tabular}{|c|c|c|}
			\hline
			Time Slot&   Success  & Failure$\vphantom{\big(^1}$ \\  \hline
			{{1{\rm st}} }  & ${N_{Cum}^1 = 0}$ & ${N_{Cum}^1 = 0}$ $\vphantom{\big(^1}$ \\ \hline
			{{2{\rm nd}} }  &  ${N_{Cum}^2 = N_{New}^1 - 1}$ & ${N_{Cum}^2 = N_{New}^1 }$ $\vphantom{\big(^1}$ \\ \hline
			{{3{\rm rd}} }  & ${N_{Cum}^3 = N_{Cum}^2 + N_{New}^2 - 1}$ & $ {N_{Cum}^3 = N_{Cum}^2 + N_{New}^2}$  $\vphantom{\big(^1}$ \\ \hline
			 \vdots  & \vdots & \vdots  \\ \hline
			 {m{\rm {th}} } & ${N_{Cum}^m = N_{Cum}^{m - 1} + N_{New}^{m-1} - 1}$ & ${N_{Cum}^m = N_{Cum}^{m - 1} + N_{New}^{m - 1}}$  $\vphantom{\big(^1}$ \\ \hline
		\end{tabular}
	}
	\label{table_accord}
\end{table}
\vspace*{-0.5cm}

\captionsetup{singlelinecheck=false} 
\vspace*{-0.5cm}
\begin{figure}[htbp!]
    \begin{center}
        \includegraphics[width=0.45\textwidth]{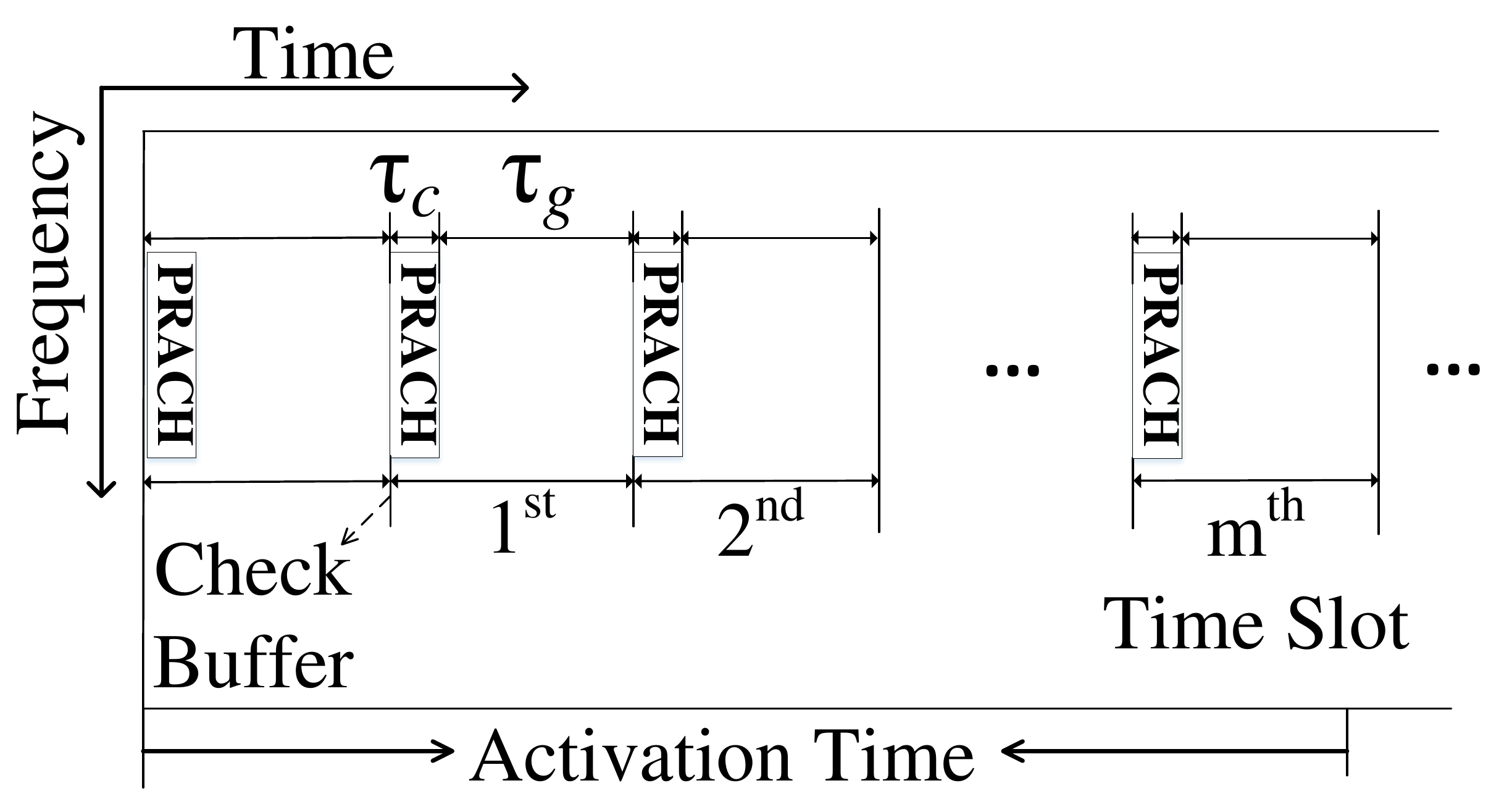}
        \vspace*{-0.4cm}
        \caption{RACH duration and gap duration, and recording the number of accumulated packets $N^m_{\rm Cum}$ in each time slot.}
        \label{fig:1}
    \end{center}
\end{figure}
\vspace*{-0.5cm}

We model the new arrived packets (${N_{New}^m}$) in the $m$th time slot at each IoT device as independent Poisson arrival process, ${\Lambda ^m_{\rm New}}$ with the same intensity $\varepsilon^m_{\rm New}$ as \cite{chen2009analysis,gotsis2012evolution,zhou2013lte} (i.e., these new packets are actually arrived within the $(m-1)$th time slot, but they are first considered in the $m$th time slot due to the slotted-Aloha behaviour). Therefore, the number of new arrival packets $N^m_{\rm New}$ in a specific time slot (i.e., within the time duration $\tau_c+\tau_g$) is described by the Poisson distribution with $N^m_{\rm New}∼\rm{Pois}(\mu ^m_{\rm New} )$, where $\mu ^m_{\rm New} = (\tau_c+\tau_g) \varepsilon^m_{\rm New} $. The accumulated packets (${N_{Cum}^m}$) at each IoT device is evolved following transmission condition over time, which is described in Table I. Specifically, a packet is removed from the buffer once the RACH succeeds, otherwise, this packet will be still in the first place of the queue, and the IoT device will try to request channel resources for the packet in the next available RACH. Note that the data transmission after a successful RACH can be easily extended following the analysis of preamble transmission success probability in RACH. Due to the main focus of this paper is analyzing the contention-based RACH in the mIoT network, we assume that the actual intended packet transmission is always successful if the corresponding RACH succeeds.

\subsection{Transmission Schemes}
\vspace*{-0.2cm}
In the cellular-based mIoT network, a huge number of IoT devices are expected to request for access frequently, such that network congestion may occur due to mass concurrent data and signaling transmission \cite{3GPP2011Study}. This network congestion can lead to a low preamble transmission success probability, and thus result in a great number of packets accumulated in buffers, which may cause unexpected delays. A possible solution is to restrict the access attempts in each IoT device according to some RACH control mechanisms. However, the efficiency of these RACH control mechanisms are required to be studied, due to that overly restricting access requests also creates unacceptable delay as well as leads to low channel resource utilization. In this paper, we study the following three schemes:
\begin{itemize}
  \item
\textbf{Baseline scheme}: each IoT device attempt RACH immediately when there exists packet in the buffer. The baseline scheme is the simplest scheme without any control of traffic. Due to RACH attempts are not be alleviated at the IoT devices, the baseline scheme can contribute to the relatively faster buffer flushing in non-overloaded network scenarios. However, once the network is overloaded, high delays and service unavailability appear due to mass simultaneous access request.
  \item
\textbf{ACB scheme}: each non-empty IoT device draws a random number $q∈ [0, 1]$, and attempts to RACH only when $q ≤ \rm P_{ACB}$, here $\rm P_{ACB}$ is the ACB factor specified by the BS according to the network condition \cite{3GPP2011Study,LTE2013dahlman}. ACB scheme is a basic congestion control method that reduces RACH attempts from the side of IoT devices based on the ACB factor. It is known that a suitable ACB factor can keep the allowable access in a reasonable density, and assure a relative high data transmission rate when the network is overloaded.
%the BS broadcasts an ACB factor $p_{ACB}∈ [0, 1]$ to all IoT devices. Each active device draws a random number $q∈ [0, 1]$. If $q ≤ p_{ACB}$, then the active device proceeds to transmit the packet; otherwise, the active device is barred for one time slot, and then, the device will try to communicate in the next time slot.
  \item
\textbf{Back-off scheme}: each non-empty IoT device transmits packets same as baseline scheme, when there exists packet in the buffer. However, when RACH fails, the IoT device automatically defers the RACH re-attempt and waits for $t_{\rm BO}$ (i.e., the BO facor specified by the BS) time slots until it trys again. Back-off scheme is another basic congestion control method, where each IoT device can automatically alleviate congestion and requires less control message from BS than that of ACB scheme \cite{hasan2013random}.
\end{itemize}

\vspace*{-0.5cm}
\subsection{Signal to Noise plus Interference Ratio}
\vspace*{-0.2cm}
As we mentioned earlier, each IoT device transmits a randomly chosen preamble to its associated BS to request for channel resources, where different preambles represent orthogonal sub-channels, and thus only IoT devices choosing same preamble have correlations. Note that IoT devices belonging to a same BS may choose same preamble, such that the intra-cell interference is considered. A preamble can be successfully received at the associated BS, if its SINR is above the threshold. Based on Slivnyak’s theorem \cite{haenggi2012stochastic}, we formulate the SINR of a typical BS located at the origin as
\vspace*{-0.2cm}
\begin{align}\label{SINR 1}
{ \mathit SINR}^m & = \frac{{\rho {h_0}}}{{ {{\cal I}_{\rm intra}} + {{\cal I}_{\rm inter}} +  {\sigma ^2}}} 
% \nonumber \\ 
% &= \frac{{\rho {h_0}}}{{\sum\limits_{{j} \in {\cal Z}_{in} }{{\mathds 1}_{\{ {N^m_{{\rm Cum}_j}} > 0\} }} {{\mathds 1}_{\{unrestricted\} }} {\rho {h_j}} + \sum\limits_{{i} \in {\cal Z}_{out} } {{{\mathds 1}_{\{ {N^m_{{\rm Cum}_i}} > 0\} }} {{\mathds 1}_{\{unrestricted\} }} {P_i}{h_i}{{\left\| {{u_i}} \right\|}^{ - \alpha }}} + {\sigma ^2}}} 
\nonumber \\ 
&  = \frac{{\rho {h_0}}}{{  \underbrace {\sum\limits_{{u_j} \in {\cal Z}_{\rm in} }{{\mathds 1}_{\{ {N^m_{{\rm New}_j}}+{N^m_{{\rm Cum}_j}} > 0\} }} {{\mathds 1}_{\{UR\} }} {\rho {h_j}} }_{{{\cal I}_{\rm intra}}}  +   \underbrace  {\sum\limits_{{u_i} \in {\cal Z}_{\rm out} } { {{\mathds 1}_{\{ {N^m_{{\rm New}_i}}+{N^m_{{\rm Cum}_i}} > 0\} }}{{\mathds 1}_{\{UR\} }} {P_i}{h_i}{{\left\| {{u_i}} \right\|}^{ - \alpha }}}   }_{{{\cal I}_{\rm inter}}} + {\sigma ^2}}},
\end{align}
where $\rho$ is the full path-loss inversion power control threshold, ${h_0}$ is the channel power gain from the typical IoT device to its associated BS, ${\sigma ^2}$ is the noise power, ${{\cal I}_{\rm intra}}$ is the aggregate intra-cell interference, ${{\cal I}_{\rm inter}}$ is the aggregate inter-cell interference\footnote{The PRACH root sequence planning is used to mitigate inter-cell interference among neighboring BS (i.e., neighboring BSs should be using different roots to generate preambles)\cite{LTE2013dahlman}. However, as \cite{gharbieh2016spatiotemporal,soorki2017stochastic}, we focus on providing a general analytical framework of mIoT network without using PRACH root sequence planning, and the extension taking into account PRACH root sequence planning can be treated in future works.%these two techniques are beyond the scope of this paper, and would not be considered in this work.
}, ${\cal Z}_{in}$ is the set of intra-cell interfering IoT devices, $N^m_{{\rm New}_j}$ is the number of new arrived packets of $j$th device in the $m$th time slot, $N^m_{{\rm Cum}_j}$ is the number of accumulated packets of $j$th device in the buffer in the $m$th time slot, ${\cal Z}_{out}$ is the set of inter-cell interfering IoT devices, ${\left\| \cdot  \right\|}$ is the Euclidean norm, ${h_i}$ is channel power gain from the $i$th inter-cell interfering IoT device to the typical BS, ${u_i}$ is the distance between the $i$th inter-cell IoT device and the typical BS, and ${P_i}$ is the actual transmit power of the $i$th inter-cell IoT device, and ${P_i}$ depends on the power control threshold $\rho$ and the distance between the $i$th inter-cell typical IoT device and its associated BS $r_i$ with $P_i{\rm{ = }}\rho {{{r_i}}^\alpha }$.

In (\ref{SINR 1}), ${{\mathds 1}_{ \{ \cdot \} }}$ is the indicator function that takes the value $1$ if the statement ${{\mathds 1}_{ \{ \cdot \} }}$ is true, and zero otherwise. Whether an IoT device generates interference depends on two conditions: 1) ${{\mathds 1}_{\{ {N^m_{\rm New}} + {N^m_{{\rm Cum}_i}} > 0\} }}$, which means that an IoT devices is able to generate interference only when its buffer is non-empty; 2) ${{\mathds 1}_{\{UR\} }}$, which means that an IoT devices is able to generate interference only when the IoT devices does not defer its access attempt due to RACH scheme. Additionally, once the two conditions are satisfied, we call the IoT device is active.
%In (\ref{INTRA 1}) and (\ref{INTER I 1 f}), whether an IoT device generates interference depends on two conditions: 1) ${{\mathds 1}_{\{ {N^m_{{\rm Cum}_i}} > 0\} }}$, which means that an IoT devices is able to generate interference only when its buffer is non-empty; 2) ${{\mathds 1}_{\{unrestricted\} }}$, which means that an IoT devices is able to generate interference only when the IoT devices does not defer its access attempt due to RACH scheme. Additionally, once the two conditions are satisfied, we call the IoT device is an active. 

Mathematically, the non-empty probability of each IoT device can be treated using the thinning process. We assume that the non-empty probability ${\cal T}^m$ and the non-restrict probability ${\cal R}^m$ of each IoT device in the $m$th time slot are defined as
\vspace*{-0.3cm}
\begin{align}\label{THINNING 1}
{\cal T}^m = {{\mathbb P} \{ {N^m_{\rm New}}+ {N^m_{\rm Cum}} > 0  \}} \text{ ,and } {\cal R}^m = {{\mathbb P} \{  unrestricted  \}}
\end{align}
where the non-restrict probability ${\cal R}^m$ depends on the RACH schemes, which will be discussed in the following. %Due to the full path-loss inversion power control, the interference from an intra-cell IoT device is strictly equal to $\rho$, and the interference from each inter-cell IoT is strictly less than $\rho$, where the received interference power depends on the distance to the typical BS. The aggregate intra-cell interference is expressed as
The main notations of this paper are summarized in Table II.

\captionsetup{singlelinecheck=true}
\vspace*{-0.3cm}
\begin{table}[htbp!]
	\centering
	\caption{Notation Table}
	{\renewcommand{\arraystretch}{0.6}
		\begin{tabular}{|*{1}{p{1.11cm}}|*{1}{p{6cm}|} |*{1}{p{1.11cm}}|*{1}{p{6cm}|} }
			\hline
			\rowcolor{Gray}
		 Notations   &    Physical means & Notations   &    Physical means $\vphantom{\big(}$ \\  \hline
			${\lambda_D}$   &  The intensity of IoT devices $\vphantom{\big(}$ & $\lambda_B$      &  The intensity of BSs $\vphantom{\big(}$ \\ 
			\rowcolor{Gray}
             $\xi  $         &   $\vphantom{\big(}$The number of available preambles are reserved for the contention-based RA *  	&		$\lambda_{Dp} $ &    The average intensity of IoT devices using the same preamble \\ 
			 $r$ & The distance $\vphantom{\big(}$ 		&	  ${\alpha}$ & The path-loss exponent $\vphantom{\big(}$ \\ 
			 \rowcolor{Gray}
			  $h$ & The Rayleigh fading channel power gain $\vphantom{\big(}$ & $P$ & The transmit power *$\vphantom{\big(}$ \\ 
			  $\rho$ & The full path-loss power control threshold *$\vphantom{\big(}$ & 	  ${\sigma^{_2}}$ & The noise power $\vphantom{\big(}$ \\ 
			  \rowcolor{Gray}
			  $\gamma_{\rm th}$ & The received SINR threshold *$\vphantom{\big(}$  & 			  ${{\cal I}_{\rm intra}}$ & The aggregate intra-cell interference $\vphantom{\big(}$ \\ 
              ${{\cal I}_{\rm inter}}$ & The aggregate inter-cell interference $\vphantom{\big(}$ & $\varepsilon_{\rm New}$ & The intensity of new arrival packets $\vphantom{\big(}$ \\ 
              \rowcolor{Gray}
			  $\tau_g$ & The gap interval duration between two RAs *$\vphantom{\big(}$ & $\tau_c$ & The duration of PRACH $\vphantom{\big(}$ \\ 
               $c$ & $c = 3.575$ is a constant & $m$ & $\vphantom{\big(}$The time slot	  \\ 
               \rowcolor{Gray}	  
			  $\mu _{\rm Cum}^m$ & $\vphantom{\big(}$The intensity of accumulated packets in the $m$th time slot  & $\mu_{\rm New}^m$ & The intensity of new arrival packets in the $m$th time slot $\vphantom{\big(}$ \\ 
			  ${\cal T}^m$ & $\vphantom{\big(}$The non-empty probability of each IoT device in the $m$th time slot  &	  ${{\cal R}^m}$ & $\vphantom{\big(}$The non-restrict probability of each IoT device in the $m$th time slot \\ 
			  \rowcolor{Gray}
			  $\vphantom{\big(}$ $Z_{D}$ & The number of active interfering IoT devices in the cell that a randomly chosen IoT device belongs to &  $Z_{B}$ & The number of active interfering IoT devices in a randomly chosen cell  \\ 
			  $\vphantom{\big(}$ $N_{New}^{m}$ & The number new arrived packets in the $m$th time slot in a specific IoT device & $N_{Cum}^{m}$ & The number accumulated packets in the $m$th time slot in a specific IoT device \\ 
			   \rowcolor{Gray}
			  ${\rm P_{ACB}}$ & The ACB factor with the ACB scheme *$\vphantom{\big(}$ & $t_{\rm BO}$ & The BO factor with the BO scheme *$\vphantom{\big(}$ \\ \hline
			  \bf{Remarks} & \multicolumn{3}{|c|}{The variables marked with * are configurable parameters.$\vphantom{\big(}$ }   \\ \hline
		\end{tabular}
	}
	\label{table_accord}
\end{table}
\vspace*{-0.0cm}

\vspace*{-0.1cm}
\section{SINR Analysis}
\vspace*{-0.2cm}
%In this section, we provide a single time slot model ($1^{st}$ time slot), where the queue status (number of packets in buffer) of each IoT device only depends on the new packets arrival process ${\Lambda ^1_{New}}$. Note that inactive IoT devices (there is no packet in buffer) do not request to access to its associated BS, such that they do not generating interference. We analyze the IoT network under the following two scenarios:

%\paragraph{ Non-collision scenario}a packet is removed from buffer when it is transmitted, and successfully recognized by the associated BS. In this case a transmission failure only occurs due to low received SINR;

%\paragraph{ Collision scenario}a packet is removed from buffer when it is transmitted, successfully recognized, and decoded by the associated BS, where both low received SINR and collision leads to a transmission failure failure.

%\subsection{Non-Collision Scenario}

In this section, we provide a general framework for the performance analysis of each single time slot and each RACH scheme. Due to that the preamble has an equal probability to be chosen, the analysis performed on a randomly chosen preamble can represent the whole network. The probability that the received SINR at the BS exceeds a certain threshold $\gamma_{\rm th}$ is written as
\vspace*{-0.3cm}
\begin{align}\label{TSP 1}
{{\mathbb P}\Big\{ {\frac{{\rho {h_o}}}{{{{\cal I}_{\rm inter}} + {{\cal I}_{\rm intra}} + {\sigma ^2}}} \ge {\gamma_{\rm th}} } \Big\}} & = {{\mathbb P}\Big\{ {h_o} \ge {\frac{\gamma_{\rm th}}{\rho } ( {{{{\cal I}_{\rm inter}} + {{\cal I}_{\rm intra}} + {\sigma ^2}}}  ) } \Big\}} 
\nonumber \\
& = {\mathbb E} \bigg[  { \rm{exp} \Big\{ - {\frac{\gamma_{\rm th}}{\rho } ( {{{{\cal I}_{\rm inter}} + {{\cal I}_{\rm intra}} + {\sigma ^2}}}  ) } \Big\}}  \bigg]
\nonumber \\
&  = {\rm exp}\big( - {\frac{{\gamma_{\rm th}}  }{\rho }{\sigma ^2}} \big){{\cal L}_{{{\cal I}_{\rm intra}}}}(\frac{{\gamma_{\rm th}}  }{\rho }){{\cal L}_{{{\cal I}_{\rm inter}}}}(\frac{{\gamma_{\rm th}}  }{\rho }) ,
\end{align}
where ${{\cal L}_{{\cal I}}}(\cdot)$ denotes the Laplace Transform of the PDF of the aggregate interference ${\cal I}$. The Laplace Transform of aggregate inter-cell interference is characterized in the following Lemma.

\vspace*{-0.1cm}
\begin{lemma}\label{lemma1}
The Laplace Transform of aggregate inter-cell interference received at the typical BS in the cellular-based mIoT network is given by
\vspace*{-0.1cm}
\begin{align}\label{INTER I F}
{{\cal L}_{{{\cal I}_{\rm inter}}}}(\frac{\gamma_{\rm th} }{\rho }) = \exp \big( { - 2{(\gamma_{\rm th})  ^{\frac{2}{\alpha }}}\frac{{{\cal T}^m}{{\cal R}^m}\lambda_{Dp} }{\lambda_B }\int_{{{(\gamma_{\rm th})}^{\frac{{ - 1}}{\alpha }}}}^\infty  {\frac{y}{{1 + {y^\alpha }}}dy} } \big) ,
\end{align}
where ${{\cal T}^m}$ and ${{\cal R}^m}$ are defined in (\ref{THINNING 1}). Remind that $\lambda_{Dp}$ is the intensity of IoT devices using same preamble.
\end{lemma}
\vspace*{-0.6cm}
\begin{proof}See Appendix \ref{INTER IF}.
\end{proof}

% Note that whether an IoT device is active only depends on the new packets arrival process ${\Lambda ^1_{\rm New}}$ in the $1$st time slot, such that the non-empty probability (the thinning factor) of each IoT device ${\cal T}^1$ in the $1$st time slot is expressed as
%\vspace*{-0.5cm}
%\begin{align}\label{THINNING P 1}
%{\cal T}^1 = {{\mathbb P}\left\{  {N^1_{\rm New}} > 0 \right\}} = 1-e^{-\mu ^1_{\rm New}} ,
%\end{align}
%where $\mu ^1_{\rm New} = \tau_g \varepsilon^1_{\rm New}$, $\varepsilon^1_{\rm New}$ is the new packets arrival rate of each IoT device in the $1$st time slot, and the probability of $N^1_{\rm New}=0$ during $\tau_g$ in the $1$st time slot is equal to $e^{-\tau_g\mu ^1_{\rm New}}$.

We perform the analysis on a randomly chosen BS and a BS associating with a randomly chosen IoT device in terms of the preamble detection probability and the preamble transmission success probability. The probability that the received SINR at a randomly chosen BS exceeds a certain threshold $\gamma_{\rm th}$ has been studied in many stochastic geometry works \cite{HElSawy2014Stochastic,novlan2013analytical,Sung2013Stochastic}. Those analyses focus on the uplink transmission channel of a cellular networks, without considering intra-cell interference due to TDMA or FDMA assumptions, and only considered inter-cell interference. In their models, the average aggregate interference is the same, no matter if the tagged BS is randomly chosen, or is determined by a randomly chosen device via association, thus the probability that the received SINR exceeds a threshold $\gamma_{\rm th}$ at a randomly chosen BS is equally same as the probability of a BS associating with a randomly chosen uplink device.

Different from the conventional stochastic geometry works in \cite{HElSawy2014Stochastic,novlan2013analytical,Sung2013Stochastic,S2013Downlink} with no intra-cell interference, we take into account the intra-cell interference due to the same preamble reuse among many IoT devices in a cell during their uplink RACH. We will derive the preamble detection probability from the view of a \textit{randomly chosen BS} (i.e., each BS has an equal probability to be chosen), and the preamble transmission success probability from the view of a BS that a \textit{randomly chosen IoT device} belongs to (i.e., the probability of a BS being chosen is determined by the number of its associated IoT devices). An example is shown in Fig. \ref{fig:2} to make a distinction between these two characteristics. For the preamble detection probability, each BS has equal probability to be chosen, and for the preamble transmission success probability, the BS 1 has a probability of 5/6 to be chosen (i.e., BS 1 covers 5 IoT devices), but BS 2 only has a probability of 1/6 to be chsoen. Concludely, the difference between these two characteristics comes from the fact that a cell, that a \textit{randomly chosen IoT device} belonging to, has chance to cover more IoT devices than a \textit{randomly chosen cell} \cite{francois2009stochastic,S2013Downlink}.

\begin{figure}[htbp!]
    \begin{minipage}[t]{0.45\textwidth}
    \centering
        \includegraphics[width=1\textwidth]{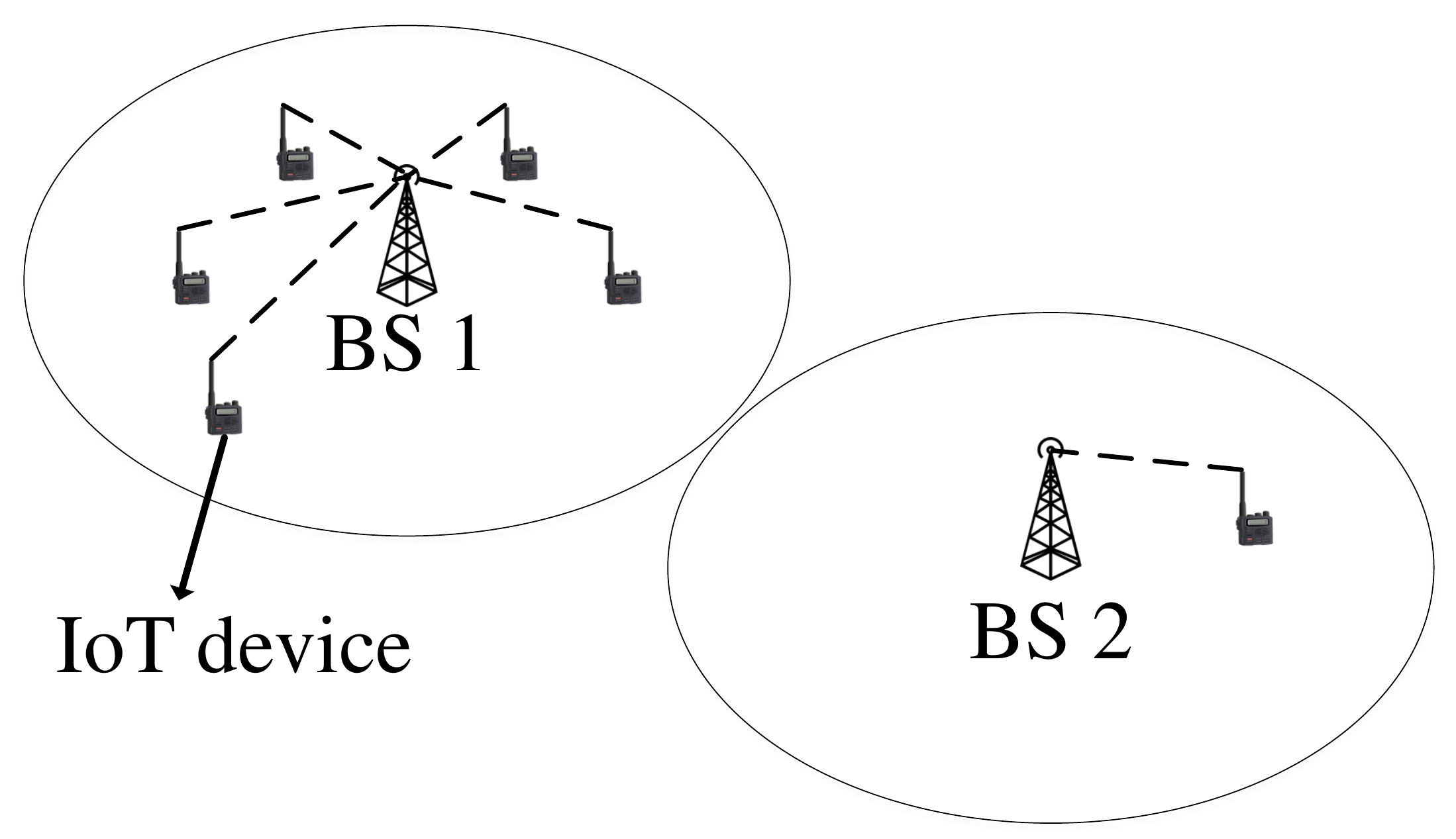}
        \vspace*{-0.5cm}
        \caption{An example of network model shows differences between the preamble detection probability and the preamble transmission success probability. %IoT devices that transmitting a same preamble at the same time to distinguish
        }
            \label{fig:2}
    \end{minipage}
    \hspace*{+1cm}
        \begin{minipage}[t]{0.45\textwidth}
    \centering
        \includegraphics[width=1\textwidth]{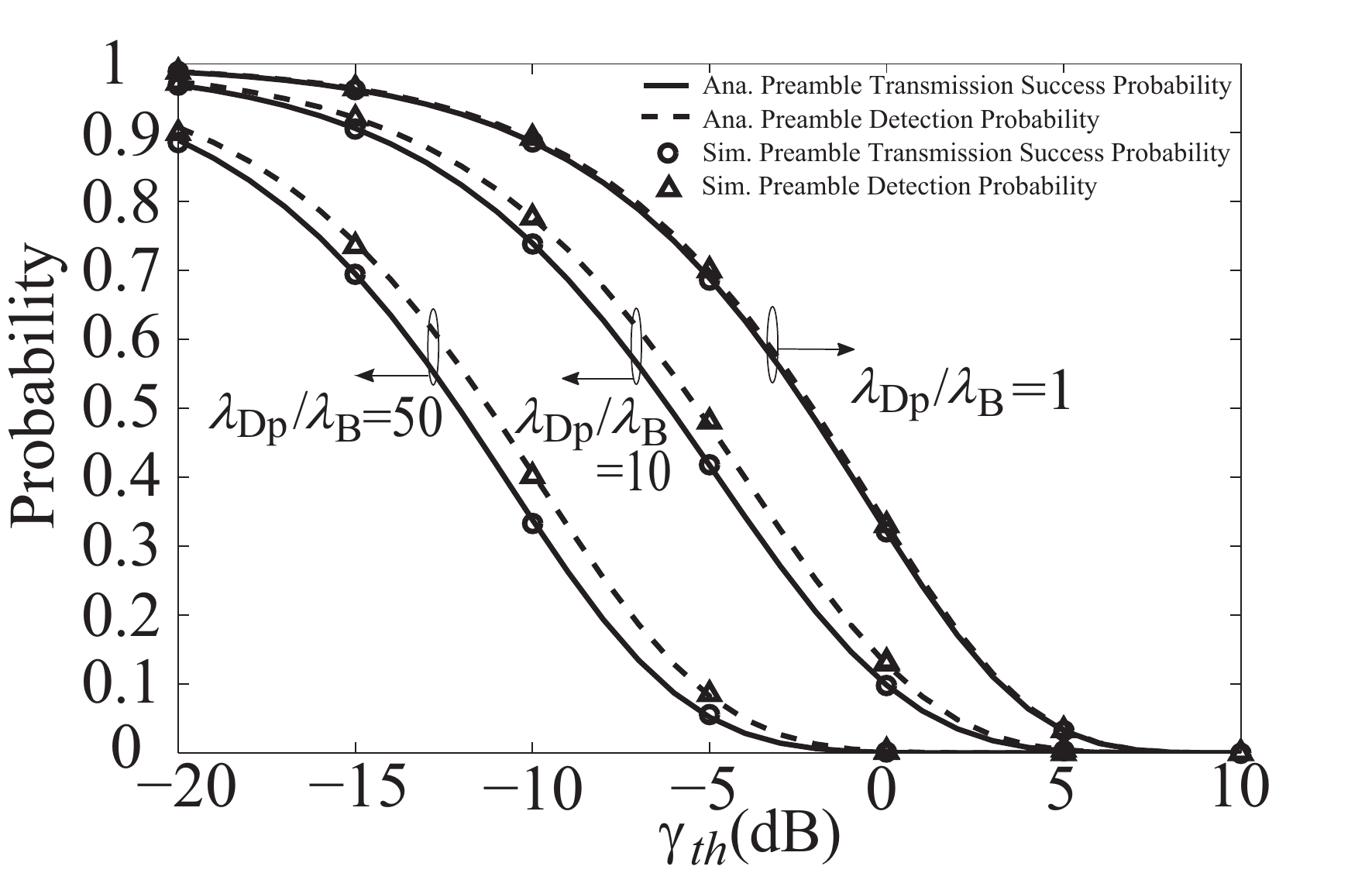}
        \vspace*{-0.5cm}
        \caption{The preamble detection probability ${\cal P}_{detection}^1$ and the preamble transmission success probability ${\cal P}^1$ versus the SINR threshold $\gamma_{\rm th}$ for the $1$st single time slot. We set ${{\cal T}^1}=1-e^{0.1}$, $\rho = -90$ dBm, ${\sigma ^2 = −90 }$ dBm, $\lambda_{B} = 10$ BS/km$^2$, $\alpha = 4$, $\gamma_{\rm th}=-10$ dB, and the baseline scheme is considered with ${{\cal R}^1}=1$.}
                \label{fig:3}
        \end{minipage}
\end{figure}
\vspace*{-0.5cm}

\vspace*{-0.2cm}
\subsection{Preamble Transmission Success Probability}
\vspace*{-0.1cm}
We first perform analysis on a BS in which a randomly chosen IoT device belongs to, where the other active IoT devices in the same cell choosing same preamble are visualized as interfering IoT devices. Since the interference generating by each intra-cell IoT device is strictly equal to $\rho$, such that the aggregate intra-cell interference only depends on the number of active interfering IoT devices in the Voronoi cell. We assume $\widehat Z_{in}$ denotes the number of active IoT device in a specific Voronoi cell, and let $Z_D=	\left| \widehat Z_{in} \right|-1$ denotes the number of active interfering IoT devices in such cell, where the Laplace Transform of aggregate intra-cell interference is conditioned on $Z_D$. The Probability Density Function (PDF) of the number of active interfering IoT devices in a Voronoi cell has been derived by the Monte Carlo method in \cite{ferenc2007size}, and conditioned on a randomly chosen IoT device in its cell, the PMF of the number of interfering intra-cell IoT devices in that cell $Z_D$ is expressed as \cite{S2013Downlink}
\vspace*{-0.2cm}
\begin{align}\label{DS chosen 1}
{{\mathbb P}} \left\{ {{Z_D} = n} \right\}{\rm{ = }} \frac{{{{\rm{c}}^{(c+1)}}\Gamma (n + c + 1){{(\frac{{{\cal T}^m}{{\cal R}^m}\lambda_{Dp}}{\lambda_B })}^{n}}}}{{\Gamma (c+1)\Gamma (n+1){{(\frac{{{{\cal T}^m}{{\cal R}^m}\lambda_{Dp}}}{\lambda_B } + c)}^{n + c + 1}}}}  ,
\end{align}
where $c = 3.575$ is a constant related to the approximate PMF of the PPP Voronoi cell, and $\Gamma \left( \cdot \right)$ is gamma function. The Laplace Transform of aggregate intra-cell interference is conditioned on the number of interfering intra-cell IoT devices ${Z_D}$, which is derived in the following Lemma.

\begin{lemma}\label{lemma2}
The Laplace Transform of aggregate intra-cell interference at the BS to which a randomly chosen IoT device belongs in the cellular-based mIoT network is given by
\vspace*{-0.1cm}
\begin{align}\label{Intra I 2}
{{\cal L}_{{{\cal I}_{\rm intra}}}}(\frac{\gamma_{\rm th} }{\rho })  & = {\mathbb P}\left\{ {{Z_D} = 0} \right\}{\rm{ + }}\sum\limits_{{\rm{n = 1}}}^\infty  {{\mathbb P}\left\{ {{Z_D} = n} \right\}{\big( {\frac{{1 }}{{1 + \gamma_{\rm th} }}} \big)^n}} ={{{\big( {1 + \frac{{{{{\cal T}^m}{{\cal R}^m}\lambda _{Dp}}{\gamma _{\rm th}}}}{{c{\lambda _B}(1 + {\gamma _{\rm th}})}}} \big)}^{ - c-1}}}.
\end{align}
\end{lemma}
\begin{proof} See Appendix \ref{Intra IF}.
\end{proof}

Substituting (\ref{INTER I F}) and (\ref{Intra I 2}) into (\ref{TSP 1}), we derive the preamble transmission success probability of the $1$st time slot $P_t^1$ in the following theorem.
\begin{theorem}\label{theorem1}
In the depicted cellular-based mIoT network, the preamble transmission success probability of a randomly chosen IoT device of the $m$st time slot is given by
\vspace*{-0.1cm}
\begin{align}\label{Pt 1}
{{\cal P}^m} =  & \exp \Big( { - \frac{{{\gamma _{\rm th}} {\sigma ^{\rm{2}}}}}{\rho } - { 2{(\gamma_{\rm th})  ^{\frac{2}{\alpha }}}\frac{{{\cal T}^m}{{\cal R}^m}\lambda_{Dp} }{\lambda_B }\int_{{{(\gamma_{\rm th})}^{\frac{{ - 1}}{\alpha }}}}^\infty  {\frac{y}{{1 + {y^\alpha }}}dy} }} \Big){{{\Big( {1 + \frac{{{{{\cal T}^m}{{\cal R}^m}\lambda _{Dp}}{\gamma _{\rm th}}}}{{c{\lambda _B}(1 + {\gamma _{\rm th}})}}} \Big)}^{ -c-1}}}.
\end{align}
%which can be approximated as (for ${\gamma _{th}} \le  0$ dB) 
%\vspace*{-0.3cm}
%\begin{align}\label{Pt 2}
% &{{\cal P}_t^1}  \approx  \exp \big( { - \frac{{\beta {\sigma ^{\rm{2}}}}}{\rho } - {{\frac{{2{\gamma _{th}}{{{{\cal T}^1} \lambda }_{Dp}}}}{{\left( {\alpha  - 2} \right){\lambda _B}}}}}{}_2{F_1}\big( {1,1 - \frac{2}{\alpha };2 - \frac{2}{\alpha }; - {\gamma _{th}}} \big)} \big)\ {{{\big( {1 + \frac{{{{{\cal T}^1}\lambda _{Dp}}{\gamma _{th}}}}{{c{\lambda _B}(1 + {\gamma _{th}})}}} \big)}^{ -c-1}}} .
%\end{align}
%Again, For the special case $\alpha = 4$, ${P_t^1}$ reduces to
%
%\vspace*{-0.3cm}
%\begin{align}\label{Pt 3}
%{{\cal P}_t^1} = & \exp \left( { - \frac{{\beta {\sigma ^{\rm{2}}}}}{\rho } - { \frac{{{\cal T}^1}\lambda _{Dp}}{\lambda_B}\sqrt {\gamma _{th}}  \arctan (\sqrt {\gamma _{th}}  )}} \right)  {{{\left( {1 + \frac{{{{{\cal T}^1}\lambda _{Dp}}{\gamma _{th}}}}{{c{\lambda _B}(1 + {\gamma _{th}})}}} \right)}^{ -c-1}}}.
%\end{align}
\end{theorem}
\begin{proof} See Appendix \ref{INTER IF} and \ref{Intra IF}.
\end{proof}

\vspace*{-0.2cm}
\subsection{Preamble Detection Probability}
\vspace*{-0.2cm}
Next, we move to the preamble decoding probability that is performed on a randomly chosen BS, and one of its active associated IoT device (with a preamble being randomly chosen) is tagged, where the other active IoT devices choosing same preamble are visualized as interfering IoT devices. Conditioned on a randomly chosen BS, the Probability Mass function (PMF) of the number of IoT devices $\left| \widehat Z_{in} \right|$ in a randomly chosen BS has been clearly introduced in \cite{S2013Downlink}, which is expressed as
\vspace*{-0.35cm}
\begin{align}\label{BS chosen 3}
{{\mathbb P}} \left\{ {{\left| \widehat Z_{in} \right|} = n} \right\}{\rm{ = }} \frac{{{{\rm{c}}^c}\Gamma (n + c){{(\frac{{{\cal T}^m}{{\cal R}^m}\lambda_{Dp}}{\lambda_B })}^{n}}}}{{\Gamma (c)\Gamma (n+1){{(\frac{{{{\cal T}^m}{{\cal R}^m}\lambda_{Dp}}}{\lambda_B } + c)}^{n + c}}}} .
\end{align}
For the Voronoi cell with at least one active IoT device, the PMF of the number of active interfering intra-cell IoT devices ${Z_B}$ in a randomly chosen Voronoi cell (BS) is given by
\vspace*{-0.05cm}
\begin{align}\label{BS chosen 4}
{{\mathbb P}} \left\{ {{Z_B} = n} \right\}{\rm{ = }}\frac{{{\mathbb P}} \left\{ {{\left| \widehat Z_{in} \right|} = n + 1} \right\}}{{1 - {{\mathbb P}} \left\{ {{\left| \widehat Z_{in} \right|} = 0} \right\}}} = \frac{{{{\rm{c}}^c}\Gamma (n + c + 1){{(\frac{{{\cal T}^m}{{\cal R}^m}\lambda_{Dp}}{\lambda_B })}^{(n+1)}}{{(1{\rm{ + }}\frac{{{{\cal T}^m}{{\cal R}^m}\lambda_{Dp}}}{{{\rm{c}}\lambda_B }})}^c}}}{{\Gamma (c)\Gamma (n+2){{(\frac{{{{\cal T}^m}{{\cal R}^m}\lambda_{Dp}}}{\lambda_B } + c)}^{n + c + 1}}\left( {{{(1{\rm{ + }}\frac{{{{\cal T}^m}{{\cal R}^m}\lambda_{Dp}}}{{{\rm{c}}\lambda_B }})}^c} - 1} \right)}}  .
\end{align}

The difference between (\ref{BS chosen 4}) and (\ref{DS chosen 1}) is clearly explained in \cite{francois2009stochastic}. Briefly speaking, in (\ref{BS chosen 4}), each Voronoi cell has an equal probability to be chosen, whilst in (\ref{DS chosen 1}), a Voronoi cell with more IoT devices has a higher probability to be chosen. Following similar approach in the proof of Lemma 2, and with the help of (\ref{BS chosen 4}), %the Laplace Transform of aggregate intra-cell interference is given by
%\vspace*{-0.1cm}
%\begin{align}\label{INTRA I 1}
% {{\cal L}_{{{\cal I}_{{\mathop{ int}} ra}}}}(\frac{\gamma_{th} }{\rho }) =\Big[ {{{\Big( {1 + \frac{{{{\cal T}^m}{{\cal R}^m}{\lambda _{Dp}}{\gamma _{th}}}}{{c{\lambda _B}(1 + {\gamma _{th}})}}} \Big)}^{ - c}} - {{\Big( {\frac{{c{\lambda _B}}}{{c{\lambda _B} + {{\cal T}^m}{{\cal R}^m}{\lambda _{Dp}}}}} \Big)}^{ - c}}} \Big]\frac{{\left( {1 + {\gamma _{th}}} \right){{\left( {1 + \left( {{{\cal T}^m}{{\cal R}^m}{\lambda _{Dp}}/ {\lambda _B}} \right)} \right)}^c}}}{{{{\big( {1 + \big( {{{\cal T}^m}{{\cal R}^m}{\lambda _{Dp}}/ {\lambda _B}} \big)} \big)}^c} - 1}}.
%\end{align}
%Substituting (\ref{INTER I F}) and (\ref{INTRA I 1}) into (\ref{TSP 1}), 
we derive the preamble detection probability of the typical BS in the $1$st time slot ${\cal P}_{detection}^m$ in the following Lemma.

\begin{lemma}\label{lemma3}
The preamble detection probability of an typical IoT device located in a randomly chosen BS in the cellular-based mIoT network is given by
\vspace*{-0.1cm}
\begin{align}\label{Pd 1}
{{\cal P}_{detection}^m} =  & \exp \Big( { - \frac{{\gamma_{th} {\sigma ^{\rm{2}}}}}{\rho } - { 2{(\gamma_{th})  ^{\frac{2}{\alpha }}}\frac{{{\cal T}^m}{{\cal R}^m}\lambda_{Dp} }{\lambda_B }\int_{{{(\gamma_{th})}^{\frac{{ - 1}}{\alpha }}}}^\infty  {\frac{y}{{1 + {y^\alpha }}}dy} }} \Big)\Big[ {{{\Big( {1 + \frac{{{{\cal T}^m}{{\cal R}^m}{\lambda _{Dp}}{\gamma _{th}}}}{{c{\lambda _B}(1 + {\gamma _{th}})}}} \Big)}^{ - c}}} \Big.
 \nonumber \\
& \Big. { - {{\Big( {\frac{{c{\lambda _B}}}{{c{\lambda _B} + {{\cal T}^m}{{\cal R}^m}{\lambda _{Dp}}}}} \Big)}^{ - c}}} \Big]\frac{{\left( {1 + {\gamma _{th}}} \right){{\left( {1 + \left( {{{\cal T}^m}{{\cal R}^m}{\lambda _{Dp}}/ {\lambda _B}} \right)} \right)}^c}}}{{{{\left( {1 + \left( {{{\cal T}^m}{{\cal R}^m}{\lambda _{Dp}}/ {\lambda _B}} \right)} \right)}^c} - 1}} .
\end{align}
\end{lemma}
\begin{proof} 
Following the proofs of Lemma \ref{lemma1} and Lemma \ref{lemma2}.
\end{proof}

In Lemma \ref{lemma3}, the preamble detection probability of an IoT device located in a \textit{randomly chosen BS} is analyzed based on the number of active interfering intra-cell IoT devices in that randomly chosen Voronoi cell (BS) in (\ref{BS chosen 4}), whereas in Theorem \ref{theorem1}, the preamble transmission success probability of a \textit{randomly chosen IoT device} is described by the number of interfering intra-cell IoT devices in that cell, where that randomly chosen IoT device belongs to in (\ref{DS chosen 1}). Fig. \ref{fig:3} plots the preamble detection probability ${\cal P}_{detection}$ and the preamble transmission success probability ${\cal P}$ versus the SINR threshold $\gamma_{\rm th}$ for a single time slot using (\ref{Pd 1}) and (\ref{Pt 1}), respectively. As expected, the preamble transmission success probability of a randomly chosen IoT device is always lower than the preamble detection probability of a randomly chosen BS, due to that a randomly chosen IoT device has higher chance to associate with a BS with large number of intra-cell interfering IoT devices as shown in (\ref{BS chosen 4}) and (\ref{DS chosen 1}), which leads to relatively low average received SINR.

In the following queue evolution analysis, we will study each packet that departs or accumulates at each IoT device in each time slot, which is determined by whether the RACH procedure succeeds or fails. To do so, the probability of RACH success in each time slot is required under the condition that each IoT device is equally treated (i.e., each IoT device has an equal probability to be chosen as a typical device no matter it is located in a cell with a relatively large or small number of IoT devices). Therefore, the following derivations are all based on the preamble transmission success probability ${{\cal P}^m}$ (i.e. it is performed on a BS in which a randomly chosen IoT device belongs to.) provided in Theorem \ref{theorem1}.

\vspace*{-0.4cm}
\section{Queue Evolution Analysis}
\vspace*{-0.2cm}

In this section, we analyze the performance of the cellular-based mIoT network in each time slot with different schemes. As mentioned in (\ref{Pt 1}), the preamble transmission success probability depends on the non-empty probability ${\cal T}^m$ and the non-restrict probability ${\cal R}^m$ of each IoT device, which raises the problem how to study the queue status of each IoT device in each time slot.

The queue status and the preamble transmission are interdependent, and imposes a causality problem. More specifically, the preamble transmission of a typical IoT device in the current time slot depends on the aggregate interference from those active IoT devices in that time slot, thus we need to know the current queue status, which is decided by the previous queue statuses, as well as the preamble transmission success probabilities of previous time slots. Recall that the evolution of queue status follows Table I, where the accumulated packets come from the packets that are not successfully transmitted in the previous time slots.

Mathematically, to derive the preamble transmission success probability of an randomly chosen IoT device in the $m$th time slot ${\cal P}^m$, we first derive the non-empty probability ${\cal T}^{m}$ and the non-restrict probability ${\cal R}^m$ of the IoT device, which are decided by ${\cal P}^{m-1}$, ${\cal T}^{m-1}$, and ${\cal R}^{m-1}$. As the number and locations of BSs and IoT devices are fixed all time once they are deployed, the locations of active IoT devices are slightly correlated across time. However, this correlation only has very little impact on the distributions of active IoT devices, and thus we approximate the distributions of non-empty IoT devices following independent PPPs in each time slot. %Next, we provide two approaches to derive the PMF and CDF of the $N^{m}_{Cum}$ of each time slot, which are probabilistic statistics and Poisson approximation.
In the rest of this section, we first describe the general analytical framework used to derive the the non-empty probability ${\cal T}^{m}$ in each time slot, and then delve into the analysis details of the non-restrict probability ${\cal R}^m$ in each time slot for each RACH scheme.

\vspace*{-0.4cm}
\subsection{Non-Empty Probability ${\cal T}^{m}$}
\vspace*{-0.2cm}
In the $1$st time slot, the number of packets in an IoT device only depends on the new packets arrival process ${\Lambda ^1_{\rm New}}$, such that the non-empty probability of each IoT device ${\cal T}^1$ in the $1$st time slot is expressed as
\vspace*{-0.3cm}
\begin{align}\label{THINNING P 1}
{\cal T}^1 = {{\mathbb P} \{ {N^1_{\rm New}} > 0 \}} = 1-e^{-\mu ^1_{\rm New}} ,
\end{align}
where $\mu ^1_{\rm New}$ is the intensity of new arrival packets. Note that the non-restrict probability in the $1$st time slot ${\cal R}^1 = 1$ with the baseline scheme, and for other RACH schemes, ${\cal R}^1$ is determined by their transmission policies, which will be detailed in the following subsection.
Substituting (\ref{THINNING P 1}) and ${\cal R}^1$ into (\ref{Pt 1}), we derive the preamble transmission success probability of a randomly chosen IoT device in the $1$st time slot ${\cal P}^1$.

Next, we derive the non-empty probability and the preamble transmission success probability of a randomly chosen IoT device in the $m$th time slot in the following Theorem.

%We have derived the preamble transmission success probability of an IoT device in the $1$st time slot given in (\ref{Pt 1}), and that of the following $M$ time slots can be derived based on the iteration process. In the $m$th time slot ($m=2,3,\cdots,M$), we first derive the number of accumulated packets $\mu_{\rm Cum}^m$, based on the proposed Poisson approximation approach (i.e., following equation (\ref{packets intensity ts2 1}). Second, we derive the non-empty probability of an IoT device ${\cal T}^{m}$ using $\mu_{\rm Cum}^m$. Finally, we derive the preamble transmission success probability of an IoT device ${\cal P}_t^{m}$ using ${\cal T}^{m}$. The preamble transmission success probability of a randomly chosen IoT device in the $m$th time slot is derived in the following Theorem.

%Based on the iteration process, we can derive the number of accumulated packets in the $m$th time slot $\mu_{\rm Cum}^m$, based on the proposed Poisson approximation approach (i.e., following equation (\ref{packets intensity ts2 1}), and then derive the non-empty probability and the preamble transmission success probability of an IoT device in the $m$th time slot. The preamble transmission success probability of a randomly chosen IoT device in $m$th time slot is derived in the following Theorem.

\vspace*{-0.3cm}
\begin{theorem}\label{theorem2}
The accumulated packets number of an IoT device in any time slot should be approximately Poisson distributed. As such, we approximate the number of accumulated packets in the $m$th time slot $N^{m}_{\rm Cum}$ as Poisson distribution $\Lambda^m_{\rm Cum}$ with intensity $\mu^m_{\rm Cum}$. The intensity of accumulated packets $\mu _{Cum}^m$ ($m>1$) in the $m$th time slot is derived as
\vspace*{-0.3cm}
\begin{align}\label{packets intensity ts3 m}
\mu _{\rm Cum}^m = 
 {\mu _{\rm New}^{m - 1} + \mu _{\rm Cum}^{m - 1} - {\cal R}^{m-1} {\cal P}^{m - 1} \big( {1 - {e^{ - \mu _{\rm New}^{m - 1} - \mu _{\rm Cum}^{m - 1}}}} \big)} .
\end{align}
The non-empty probability of each IoT device in the $m$th time slot is derived as
\vspace*{-0.3cm}
\begin{align}\label{ACTIVE P m}
{\cal T}^{m} = 1 - {{{e}}^{ - \mu _{\rm New}^m - \mu _{\rm Cum}^m}}.
\end{align} 
Substituting ${\cal T}^{m}$ and ${\cal R}^m$ into (\ref{Pt 1}), we derive the preamble transmission success probability of a randomly chosen IoT device in the $m$th time slot ${\cal P}^m$. Note that ${\cal R}^m =1$ with the baseline scheme, and for other RACH schemes, ${\cal R}^m$ are determined by their transmission policies, which will be detailed in the following subsection\footnote{With minor modification, this theorem can also be leveraged to study other traffic models, such as the time limited Uniform Distribution and the time limited Beta distribution \cite{3GPP2011Study}.}.
%The preamble transmission success probability of a randomly chosen IoT device in the $m$th time slot in the spatio-temporal model is derived as
%\vspace*{-0.2cm}
%\begin{align}\label{Pt ts m 1}
%{{\cal P}_t^m} =  & \exp \left( { - \frac{{{\gamma _{\rm th}} {\sigma ^{\rm{2}}}}}{\rho } - { 2{(\gamma_{\rm th})  ^{\frac{2}{\alpha }}}\frac{{{\cal T}^m}\lambda_{Dp} }{\lambda_B }\int_{{{(\gamma_{\rm th})}^{\frac{{ - 1}}{\alpha }}}}^\infty  {\frac{y}{{1 + {y^\alpha }}}dy} }} \right){{{\left( {1 + \frac{{{{{\cal T}^m}\lambda _{Dp}}{\gamma _{th}}}}{{c{\lambda _B}(1 + {\gamma _{\rm th}})}}} \right)}^{ -c-1}}}.
%\end{align}
\end{theorem}

%In the $1^{st}$ time slot, the density of interfering IoT devices $\widehat{\lambda}^1_{Dp}$ obtained by the thinning factor $1-e^{-\mu^1_n}$ in (\ref{THINNING P 1}), which can be visualised as the non-empty probability of each IoT device in the $1^{st}$ time slot. Let us assume $ {\cal T}^{m} $ denotes the non-empty probability of $m^{th}$ time slot, which is jointly impacted by the number of new arrived packets ($N^m_n$) and accumulated packets ($N^{m}_b$) in buffer. The new arrival packets is denoted by the Poisson arrival process $\Lambda^m_n$, where its Cumulative Distribution Function (CDF) could be easily derived, and will not be represented.

\begin{proof}
We first derive the non-empty probability in each time slot using exact probabilistic statistics. In the $2$nd time slot, the PMF of the accumulated packets $N^{2}_{Cum}$ is expressed as
\vspace*{-0.2cm}
\begin{align}\label{packets pmf ts 2 1}
f_{N^{2}_{\rm Cum}}(x) = \left\{
\begin{aligned}
    & {e^{ - \mu _{\rm New}^1}} + \mu _{\rm New}^1{e^{ - \mu _{\rm New}^1}}{\cal R}^{1}{\cal P}^1  , & x = 0,  \\
    & \frac{{{(\mu _{\rm New}^1)}^x{e^{ - \mu _{\rm New}^1}}}}{{x!}}(1 - {\cal R}^{1}{\cal P}^1) + \frac{{(\mu {_{\rm New}^1})^{x + 1}{e^{ - \mu _{\rm New}^1}}}}{{(x + 1)!}}{\cal R}^{1}{\cal P}^1   , & x > 0 .  \\
\end{aligned} \right.
\end{align}
The reason for (\ref{packets pmf ts 2 1}) is that the number of accumulated packets in the $2$nd time slot equals to $x$ occurs only when 1) the number of accumulated packets in the $1$st time slot equals to $x+1$, and one packet is successfully transmitted in the $1$st time slot, and 2) the number of accumulated packets in the $1$st time slot equals to $x$, and no packet is successfully transmitted in the $1$st time slot.

Based on (\ref{packets pmf ts 2 1}), we derive the CDF of the number of accumulated packets in the $2$nd time slot $N^{2}_{\rm Cum}$ as
\vspace*{-0.4cm}
\begin{align}\label{packets cdf ts2 1}
F_{N^{2}_{\rm Cum}}(y) = \sum\limits_{x = 0}^y { f_{N^{2}_{\rm Cum}}(x) } = \frac{{{{(\mu _{\rm New}^1)}^{y + 1}}{e^{ - \mu _{\rm New}^1}}}}{{(y + 1)!}}{\cal R}^{1}{{\cal P}^1} +  \sum\limits_{x = 0}^y { {\frac{{{{(\mu _{\rm New}^1)}^x}{e^{ - \mu _{\rm New}^1}}}}{{x!}}} } .
\end{align}
We are interested in the zero-accumulated packets probability in the $2$nd time slot, since it determines the density of non-empty IoT devices (with more than one packet in the buffer) in that time slot, and the activity probability of IoT devices. Based on the probabilistic statistics and (\ref{packets pmf ts 2 1}), we present the non-empty probability of IoT devices in the 2$nd$ time slot as
\vspace*{-0.3cm}
\begin{align}\label{ACTIVE P 1}
{\cal T}_{\rm BL}^{2} =  1 - {{{e}}^{ - {\mu _{\rm New}^2}}}  \big( {{e}^{ - {\mu _{\rm New}^1}}} + {\mu _{\rm New}^1}{{e}^{ - {\mu _{\rm New}^1}}}{\cal R}^{1}{{\cal P}^1}\big) .
\end{align}
Substituting (\ref{ACTIVE P 1}) and ${\cal R}^2=1$ into (\ref{Pt 1}), we derive the preamble transmission success probability of a randomly chosen IoT device in the $2$nd time slot ${\cal P}^2$.

%Following the proof of Theorem 2 and substituting (\ref{ACTIVE P 1}) into (\ref{INTER I 3}), we derive the preamble transmission success probability of a randomly chosen IoT device in the $2$nd time slot as
%and (\ref{Intra IF 1})
%\vspace*{-0.3cm}
%\begin{align}\label{Pt ts2}
%{{\cal P}_t^2} =  & \exp \Big( { - \frac{{\gamma_{\rm th} {\sigma ^{\rm{2}}}}}{\rho } - { 2{(\gamma_{\rm th})  ^{\frac{2}{\alpha }}}\frac{{{\cal T}^2}\lambda_{Dp} }{\lambda_B }\int_{{{(\gamma_{\rm th})}^{\frac{{ - 1}}{\alpha }}}}^\infty  {\frac{y}{{1 + {y^\alpha }}}dy} }} \Big){{{\Big( {1 + \frac{{{{{\cal T}^2}\lambda _{Dp}}{\gamma _{th}}}}{{c{\lambda _B}(1 + {\gamma _{\rm th}})}}} \Big)}^{ -c-1}}}.
%\end{align}

Similar as (\ref{packets pmf ts 2 1}) and (\ref{packets cdf ts2 1}), we can derive the PMF and the CDF of the number of accumulated packets in the $3$rd time slot $N^{3}_{\rm Cum}$ as
\vspace*{-0.3cm}
\begin{align}\label{packets pmf ts3 1}
&{f_{N_{\rm Cum}^3}}(x) = \nonumber\\&
 \left\{
\begin{aligned}
&  {e^{ - \mu _{\rm New}^2}}{f_{N_{\rm Cum}^2}}(0) + {\cal R}^{2}{\cal P}^2\left[ {\mu _{New}^2{e^{ - \mu _{\rm New}^2}}{f_{N_{\rm Cum}^2}}(0) + {e^{ - \mu _{\rm New}^2}}{f_{N_{\rm Cum}^2}}(1)} \right]   , & x = 0,  \\
   &  ({1 - {\cal R}^{2}{\cal P}^2})\sum\limits_{z = 0}^x {\Big[ {\frac{{{{\left( {\mu _{\rm New}^2} \right)}^{z}}{e^{ - \mu _{\rm New}^2}}}}{{\left( {z} \right)!}}{f_{N_{\rm Cum}^2}}(x - z)} \Big]} 
   \\
   & +  {\cal R}^{2} {\cal P}^2\sum\limits_{z = 0}^{x + 1} {\Big[ {\frac{{{{\left( {\mu _{\rm New}^2} \right)}^{z }}{e^{ - \mu _{\rm New}^2}}}}{{\left( {z} \right)!}}{f_{N_{\rm Cum}^2}}(x + 1 - z)} \Big]}     , & x > 0 ,
\end{aligned} \right.
\end{align}
and
\vspace*{-0.3cm}
\begin{align}\label{packets cdf ts3 1}
\hspace*{-0.1cm} {F_{N_{\rm Cum}^3}}(y) \hspace*{-0.1cm} = \hspace*{-0.1cm} {\cal R}^{2} {\cal P}^2\sum\limits_{z = 0}^{y + 1} {\Big[ {\frac{{{{\left( {\mu _{\rm New}^2} \right)}^z}{e^{ - \mu _{\rm New}^2}}}}{{\left( z \right)!}}{f_{N_{\rm Cum}^2}}(y + 1 - z)} \Big]} \hspace*{-0.1cm} 
+\hspace*{-0.1cm}  \sum\limits_{x = 0}^y {\sum\limits_{z = 0}^x {\Big[ {\frac{{{{\left( {\mu _{\rm New}^2} \right)}^z}{e^{ - \mu _{\rm New}^2}}}}{{\left( z \right)!}}{f_{N_{\rm Cum}^2}}(x - z)} \Big]} } ,
\end{align}
respectively. In (\ref{packets pmf ts3 1}) and (\ref{packets cdf ts3 1}), ${f_{N_{\rm Cum}^2}}(x)$ is given in (\ref{packets pmf ts 2 1}). Generally, the PMF and CDF of $N^{m}_{\rm Cum}$ in the $m$th time slot can be derived by the iteration process. 

% To derive the distribution of the number of packets in the $m$th ($m>0$) time slot, we approximate the number of packets of each IoT device result from the previous accumulated packets ($N^{m}_b$), which is blocked due to a high queue length and/or bad channel condition, is denoted by independent PP $\Lambda^m_b$, with same intensity $\mu^m_b$. Without loss of generality, we assume $\Lambda^m_b$ is interdependent with all previous new packets arrival PPs and transmission success probabilities, and is independent with the new packets arrival PP $\Lambda^m_n$ of $m$th time slot.

However, as $m$ increases, the complexity of these derivations exponentially increases, and thus they become hard to analyze. Due to the new packets arrival at each IoT device is modeled by independent Poisson process, the packets departure can be treated as an approximated thinning process (i.e., the thinning factor is a function relating to the preamble transmission success probability, the non-empty probability, and the non-restrict probability) of the arrived packets. Therefore, after this thinning process in a specific time slot, the least packets (i.e. the accumulated packets) number at each IoT device can be approximated as Poisson distribution with the same mean. As such, we approximate the number of accumulated packets in the $m$th time slot $N^{m}_{\rm Cum}$ as Poisson distribution $\Lambda^m_{\rm Cum}$ with intensity $\mu^m_{\rm Cum}$. The intensity of accumulated packets $\mu _{Cum}^m$ ($m>1$) in the $m$th time slot is derived as

 As such, we approximate the number of accumulated packets in the $m$th time slot as a Poisson distribution ($m>1$), 
%which largely simplifies the derivations. We approximate
where the number of accumulated packets of an IoT device in the $m$th time slot $N^{m}_{\rm Cum}$ is approximated as Poisson distribution $\Lambda^m_{\rm Cum}$ with intensity $\mu^m_{\rm Cum}$. 
%The intensity of the accumulated packets in the $1$st time slot $\mu^1_{\rm Cum}$ is equal to zero, due to the buffer of each IoT device is set as empty at the beginning of the first time slot ($N^{1}_b = 0$). 
In the $2$nd time slot, $\mu^2_{\rm Cum}$ depends on the new packets arrival rate $\mu^1_{\rm New}$ and the preamble transmission success probability ${\cal P}^1$ of an IoT device in the $1$st time slot, which is given by
\vspace*{-0.3cm}
\begin{align} \label{packets intensity ts2 1}
\mu _{\rm Cum}^2 & =  \underbrace { {\cal R}^{1} {\cal P}^1\big( {\sum\limits_{x = 1}^\infty  {{f_{N_{\rm New}^1}}(x) \cdot (x - 1)} } \big)}_{(a)} + \underbrace { \big( (1- {\cal R}^{1}) + {\cal R}^{1} ( {1 - {\cal P}^1} ) \big){\big( {\sum\limits_{x = 1}^\infty  {{f_{N_{\rm New}^1}}(x) \cdot x} } \big)}}_{(b)}
 \nonumber \\
 &  = {\cal R}^{1} {\cal P}^1\big( {\sum\limits_{x = 1}^\infty  {\frac{{{{\left( {\mu _{\rm New}^1} \right)}^x}{e^{ - \mu _{\rm New}^1}}(x - 1)}}{{x!}}} } \big) + ( {1 - {\cal R}^{1} {\cal P}^1}) \mu _{\rm New}^1
\nonumber \\
&  = {\cal R}^{1} {\cal P}^1\big( {\sum\limits_{x = 0}^\infty  {\frac{{{{\left( {\mu _{\rm New}^1} \right)}^x}{e^{ - \mu _{\rm New}^1}}x}}{{x!}} - \sum\limits_{x = 1}^\infty  {\frac{{{{\left( {\mu _{\rm New}^1} \right)}^x}{e^{ - \mu _{\rm New}^1}}}}{{x!}}} } } \big) + ( {1 - {\cal R}^{1} {\cal P}^1}) \mu _{\rm New}^1
\nonumber\\
&  =  {\mu _{\rm New}^1 - {\cal R}^{1}{\cal P}^1\big( {1 - {e^{ - \mu _{\rm New}^1}}} \big)} ,
\end{align}
where $\mu ^1_{\rm New} = (\tau_c + \tau_g)\varepsilon^1_{\rm New}$, $\varepsilon^1_{\rm New}$ is the new packets arrival rate of each device in the $1$st time slot, ${f_{N_{\rm New}^1}}(\cdot)$ is the PMF of the number of new arrived packets $N_{\rm New}^1$, ${\cal P}^1$ is given in (\ref{Pt 1}) of Theorem 1. In (\ref{packets intensity ts2 1}), $(a)$ is the density of the accumulated packets in the $2$nd time slot when a packet is successfully transmitted in the $1$st time slot, and $(b)$ is the density of the accumulated packets in the $2$nd time slot when the congestion alleviation or the unsuccess transmission occurs in the $1$st time slot.

According to Poisson approximation and (\ref{packets intensity ts2 1}), the CDF of the number of packets in the $2$nd time slot due to previous accumulated packets $N_{\rm Cum}^2$ is approximated as
\vspace*{-0.2cm}
\begin{align}\label{packets CDF g 1}
\hspace{-0.3cm} 
{F_{N_{\rm Cum}^2}}(y) &\approx \sum\limits_{z = 0}^y {\frac{{{{\left( {\mu _{\rm Cum}^2} \right)}^z}{e^{ - \mu _{\rm Cum}^2}}}}{{z!}}}  = \sum\limits_{z = 0}^y {\frac{{{{ {\big( {\mu _{\rm New}^1 - {\cal R}^{1} {\cal P}^1\big( {1 - {e^{ - \mu _{\rm New}^1}}} \big)} \big)} }^z}{e^{ -  {\mu _{\rm New}^1 - {\cal R}^{1} {\cal P}^1\big( {1 - {e^{ - \mu _{\rm New}^1}}} \big)} }}}}{{z!}}},
\end{align}
and the non-empty probability of an IoT devices in the $2$nd time slot is approximated as
\vspace*{-0.3cm}
\begin{align}\label{ACTIVE P 2}
{\cal T}^{2} \approx 1 - {{{e}}^{ - \mu _{\rm New}^2 - \mu _{\rm Cum}^2}},
\end{align}
where $\mu _{\rm Cum}^2$ is given in (\ref{packets intensity ts2 1}).

Similarly, the intensity of the number of accumulated packets in the $3$rd time slot $\mu^3_{\rm Cum}$ is
\vspace*{-0.3cm}
\begin{align}\label{packets intensity ts3 1}
 \mu _{\rm Cum}^3 &   = {\mu _{\rm New}^2 + \mu _{\rm Cum}^2 - {\cal R}^{2}{\cal P}^2\big( {1 - {e^{ - \mu _{\rm New}^2 - \mu _{\rm Cum}^2}}} \big)} ,
\end{align}
where $\mu _{\rm Cum}^3$ is given in (\ref{packets intensity ts3 1}). Thus, we approximate the CDF of the number of accumulated packets in the $3$nd time slot $N_{\rm Cum}^3$ as
\vspace*{-0.2cm}
\begin{align}\label{packets CDF g 2}
{F_{N_{\rm Cum}^3}}(y) \approx \sum\limits_{z = 0}^y {\frac{{{{\left( {\mu _{\rm Cum}^3} \right)}^z}{e^{ - \mu _{\rm Cum}^3}}}}{{z!}}} .
\end{align}
The intensity of the number of accumulated packets in the $m$th time slot ($m>3$) is derived following (\ref{packets intensity ts2 1}), which is already given in (\ref{packets intensity ts3 m}). For simplicity, we omit this expression here.

 \vspace*{-0.6cm}
 \captionsetup{singlelinecheck=false} 
\begin{figure}[H]      
    \begin{center}
        \includegraphics[width=0.55\textwidth]{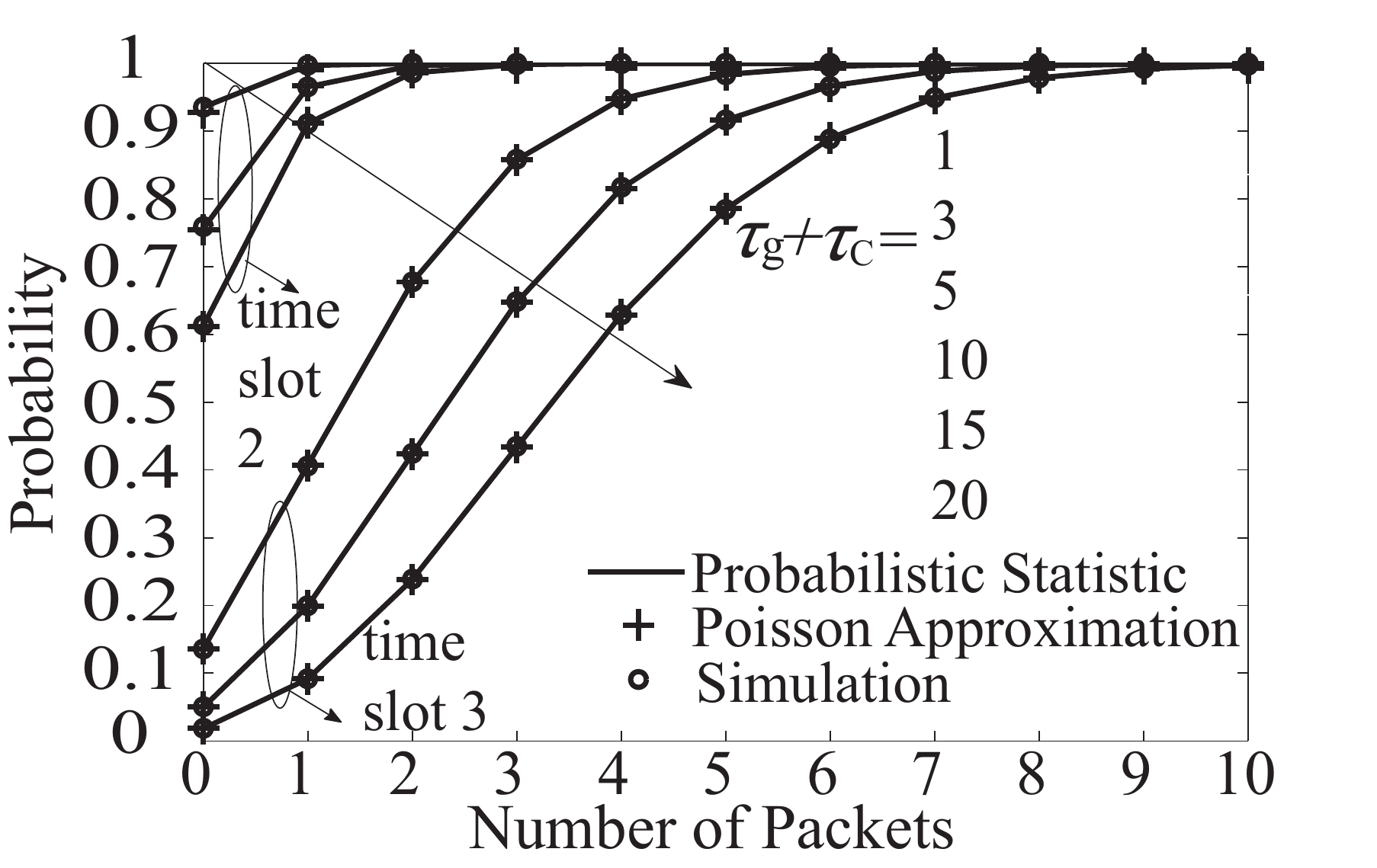}
        \vspace*{-0.3cm}
        \caption{Comparing the CDFs of the number of accumulated packets between probabilistic statistics and Poisson approximation in the $2$nd and the $3$rd time slots. We present 6 scenarios with different RACH interval durations, where $(\tau_c+\tau_g) = 1,3,5,10,15$ and $20$ ms. The simulation parameters are $\lambda_{B} = 10$ BS/km$^2$, $\lambda_{Dp} = 100$ IoT deivces/preamble/km$^2$, $\rho = -90$ dBm, ${\sigma ^2 = −90 }$ dBm, $\varepsilon^1_{\rm New}=\varepsilon^2_{\rm New}=\varepsilon^3_{\rm New}=0.1$ packets/ms, and the baseline scheme with ${\cal R}^m =1$.}
   \label{fig:4}    
    \end{center}
\end{figure}
 \vspace*{-0.6cm}
 
% The non-empty probability of an IoT device in the $3$rd time slot ${\cal T}^{3}$ under the Poisson approximation is expressed as
%\vspace*{-0.3cm}
%\begin{align}\label{ACTIVE P 3}
%{\cal T}^{3} \approx 1 - {{{e}}^{ - \mu _{\rm New}^3 - \mu _{\rm Cum}^3}}.
%\end{align}

Fig. \ref{fig:4} shows the CDFs of the number of accumulated packets via simulation, as well as calculating by the probabilistic statistics and the Poisson approximation. We see the close match among the probabilistic statistics, Poisson approximation and the simulation results, which validates our approximation approach. More simulation results will be provided in the Section V to validate the Poisson approximation approach. 

\end{proof}

\vspace*{-0.4cm}
\subsection{Non-Restrict Probability ${\cal R}^{m}$}
\vspace*{-0.2cm}
\subsubsection{The Baseline Scheme} The baseline scheme allows each IoT device to attempt RACH immediately when there exists packet in the buffer, and thus the non-restrict probability is always equal to 1 in any time slot (${\cal R}_{\rm BL}^m = 1$). 
\subsubsection{The ACB Scheme} In the ACB scheme, the BS first broadcasts the ACB factor $\rm P_{ACB}$, then each non-empty IoT device draws a random number $q∈ [0, 1]$, and attempts to RACH only when $q$ smaller than or equal to the ACB factor $\rm P_{ACB}$. Therefore, the non-restrict probability is always equal to $\rm P_{ACB}$ in any time slot (${\cal R}_{\rm ACB}^m = {\rm P_{ACB}}$). 
\subsubsection{The Back-Off Scheme} In the back-off scheme, each IoT device defers its access and waits for $t_{\rm BO}$ time slots, when such IoT devices failed to transmit a packet in the last time slot. The analysis of the non-restrict probability with the back-off scheme ${\cal R}_{\rm BO}^m$ is similar to the ACB scheme, due to the back-off procedure can be visualised as a group of IoT devices are completely barred in a specific time slot. In the $1$st time slot, none of IoT device defers the access attempt, such that the transmission procedure is same as the baseline scheme (${\cal R}_{\rm BL}^1 = 1$). After the $1$st time slot, the back-off procedure starts to execute, an non-empty IoT device defers its access attempt if the back-off being trigged. %The preamble transmission success probability of a randomly chosen IoT device with the back-off scheme in the $m$th time slot is also derived based on the iteration process, which is presented in the following Proposition.

Due to the back-off mechanism, only active IoT devices without RACH attempt failures in the last $t_{\rm BO}$ time slots can attempt to transmit a preamble, and only those IoT devices generate interference that determine the preamble transmission success probability in the $m$th time slot. The non-restrict probability with the back-off scheme ${\cal R}_{\rm BO}^m$ is derived as
\vspace*{-0.2cm}
\begin{align}\label{AR BO 1}
{\cal R}_{\rm BO}^m = \left\{
\begin{aligned}
&1 , &m=1,\\
& 1-{{\Big[ {\sum\limits_{j = 1}^{m - 1}  \underbrace {(1 - {\cal P}_{\rm BO}^{j}){\cal T}_{\rm BO}^{j} {{\cal R}_{\rm BO}^{j}}}_{(a)} } \Big]} {{\cal T}_{\rm BO}^m}}
,&\left( {t_{\rm BO} + 1} \right) \ge m > 1,  \\
& 1-{{\Big[ \sum\limits_{j = m - t_{\rm BO}}^{m - 1} \underbrace{(1 - {\cal P}_{\rm BO}^{j}){\cal T}_{\rm BO}^{j}{{\cal R}_{\rm BO}^{j}}}_{(a)}  \Big]} {{\cal T}_{\rm BO}^m}} ,& m > t_{\rm BO},  \\
\end{aligned}\right.
\end{align}
where $(a)$ is the probability that an randomly chosen IoT device fails to transmit a preamble in the $j$th time slot, and thus this IoT device would defer its RACH request in the $m$th time slot due to the back-off mechanism.

\section{Performance Metrics}
We have derived the preamble transmission success probability in each time slot in the last section, and then based on the derived probability, many performance metrics can be obtained.
\vspace*{-0.1cm}
\subsection{The Number of Received Packets per BS}
\vspace*{-0.2cm}
We first analyze the number of received packets per BS of cellular-based mIoT networks as a function of the densities of IoT devices using same preamble and BS, which reflects the density of successfully RACH IoT devices using same preamble per BS (\cite{S2013Downlink}, {e.q. (6)}). In our model, the number of received packets per BS in the $m$th time slot ${\cal C}^m$ is defined as
\vspace*{-0.2cm}
\begin{align}\label{Ct 1}
{\cal C}^m  \buildrel \Delta \over =  {{{{\cal T}^m}{{\cal R}^m}{\lambda _{Dp}}}}{{}} \cdot {{\cal P}^m} / {\lambda _B}.
\end{align}
Substituting (\ref{Pt 1}) into (\ref{Ct 1}), the number of received packets per BS ${\cal C}^m$ is derived as
\vspace*{-0.1cm}
\begin{align}\label{Ct 2}
 &{{\cal C}^m} \hspace*{-0.1cm} = \hspace*{-0.1cm} \frac{{{{\cal T}^m}{{\cal R}^m}{\lambda _{Dp}}}}{{{\lambda _B}}} \exp \Big( { - \frac{{{\gamma _{\rm th}} {\sigma ^{\rm{2}}}}}{\rho } - { 2{(\gamma_{\rm th})  ^{\frac{2}{\alpha }}}\frac{{{\cal T}^m}{{\cal R}^m}\lambda_{Dp} }{\lambda_B }\int_{{{(\gamma_{\rm th})}^{\frac{{ - 1}}{\alpha }}}}^\infty  {\frac{y}{{1 + {y^\alpha }}}dy} }} \Big)\ {{{\Big( {1 + \frac{{{{{\cal T}^m}{{\cal R}^m}\lambda _{Dp}}{\gamma _{\rm th}}}}{{c{\lambda _B}(1 + {\gamma _{\rm th}})}}} \Big)}^{ -c-1}}} \hspace*{-0.2cm} .
\end{align}
In (\ref{Ct 1}), ${\cal C}^m$ is negatively proportional to the density ratio ${\lambda_B }$, and in (\ref{Pt 1}), the preamble transmission success probability ${\cal P}^m$ is positively proportional to the density ratio ${\lambda_B }$, which positively improves ${\cal C}^m$. Therefore, ${\lambda_B }$ introduce a tradeoff in the system performance of ${\cal C}^m$, which is jointly determined by two opposite factors: 1) the average received SINR of each BS, 2) the average number of associated IoT devices of each BS. Practically, when BSs are deployed with a relatively large density, rare IoT devices can successfully transmit a preamble to their associated BSs, due to the large interference leading to extremely low received SINR. In this scenario, ${\cal C}^m$ is dominantly determined by the factor 1 (i.e., average received SINR of each BS), and thus increasing the BS intensity $\lambda_B$ can greatly improve the number of received packets per BS. However, increasing the BS intensity increases the received SINR, but decreases the average number of associated IoT devices, which contributes to higher number of received packets per BS in the scenario of overloaded network, but decreases the that in the scenario of non-overloaded network due to low utilization of channel resources (i.e., the factor 2 dominantly determined ${\cal C}^m$ in this scenario). Therefore, ${\cal C}^m$ is concave downward, and there exists a optimal BS density deployment which enables the maximum number of received packets per BS as shown in (\ref{Ct max 1}).

To obtain the optimal number of received packets per BS in proposed IoT-enabled cellular network, we take the first derivative on ${C^m}$, and obtain the density of BSs achieving the maximum number of received packets per BS $\lambda^*_{B}$ as
\vspace*{-0.1cm}
\begin{align}\label{Ct max 1}
 {\lambda^*_{B}}  =& \frac{{{{{{\cal T}^m}{{\cal R}^m} \lambda }_{Dp}}}}{2}\Bigg( {{ 2{(\gamma_{\rm th})  ^{\frac{2}{\alpha }}}\int_{{{(\gamma_{\rm th})}^{\frac{{ - 1}}{\alpha }}}}^\infty  {\frac{y}{{1 + {y^\alpha }}}dy} } +  \frac{{{\gamma _{\rm th}}}}{{\left( {1 + {\gamma _{\rm th}}} \right)}} + } \Bigg.
 \nonumber \\
 &\Bigg. \hspace*{-0.1cm} {\sqrt {{ \hspace*{-0.1cm} \Big( \hspace*{-0.1cm} { 2{(\gamma_{\rm th})  ^{\frac{2}{\alpha }}} \hspace*{-0.2cm} \int_{{{(\gamma_{\rm th})}^{\frac{{ - 1}}{\alpha }}}}^\infty  {\frac{y}{{1 +  {y^\alpha }}}dy} } \Big) ^2} \hspace*{-0.1cm} + \hspace*{-0.1cm} \frac{{{\gamma _{\rm th}}^2}}{{{{\left( {1 + {\gamma _{\rm th}}} \right)}^2}}} \hspace*{-0.1cm} + \hspace*{-0.1cm} ( {4 + \frac{8}{c}} ) { \Big( {\int_{{{(\gamma_{\rm th})}^{\frac{{ - 1}}{\alpha }}}}^\infty  {\frac{y}{{1 + {y^\alpha }}}dy} } \Big)}\frac{{{(\gamma_{\rm th})  ^{\frac{\alpha+2}{\alpha }}}}}{{\left( {1 + {\gamma _{\rm th}}} \right)}}} } \Bigg) .
\end{align}

\subsection{Mean of ${\cal C}^m$ and ${\cal P}^m$}

The number of received packets per BS in the $m$th time slot ${\cal C}^m$ is derived by using ${\cal P}^m$ following (\ref{Ct 2}). Next, we derive the mean of preamble transmission success probabilities of a randomly chosen IoT device over $M$ time slots and the mean of number of received packets per BS over $M$ time slots, which are expressed as
\vspace*{-0.15cm}
\begin{align}\label{aPt m}
{\mathbb E}[{\cal P}^m] = \big( {\sum\limits_{m = 1}^M {\cal P}^m} \big)/{M} \text{ ,and }  {\mathbb E}[{\cal C}^m] \hspace*{-0.1cm} = \hspace*{-0.1cm} \big( {\sum\limits_{m = 1}^M {\cal C}^m}\big) /{M} .
\end{align}

\vspace*{-1.1cm}
\subsection{Average Queue Length}
\vspace*{-0.1cm}
The preamble transmission success probability provides insights on the received SINR for a random IoT device in each time slot, but does not evaluate the packets accumulation status. Many previous works have indicated that the queue length is a good indication of network congestion \cite{hasan2013random,laya2014random}. The queue length refers to the number of packets that are waiting in buffer to be transmitted \cite{alfa2010queueing}. Next, we evaluate the average queue length ${\mathbb E}[Q^m]$, which denotes the average number of packets accumulated in the buffer in the $m$th time slot, which is derived as
\vspace*{-0.25cm}
\begin{align}\label{t 1 m 71}
& {\mathbb E}[Q^m] = \mu _{\rm New}^{m} + \mu _{\rm Cum}^{m} - {{\cal R}}^{m} {\cal T}^{m}{\cal P}^{m} ,
\end{align}
where $\mu _{\rm Cum}^{m}$ is the intensity of number of accumulated packets in the $m$th time slot given in (\ref{packets intensity ts3 m}), $\mu _{\rm New}^{m}$ is the intensity of the new arrival packets in the $m$th time slot, ${\cal P}^{m} $ is given in Theorem \ref{theorem1}, ${\cal T}^{m}$ is given in Theorem \ref{theorem2}, and ${{\cal R}}^{m} $ is given in Section IV.B.

\section{Numerical Results}

In this section, we validate our analysis via independent system level simulations, where the BSs and IoT devices are deployed via independent PPPs in a 100 km$^2$ area. Each IoT device employs the channel inversion power control, and associated with its nearest BS. Importantly, the real buffer at each IoT device is simulated to capture the packets arrival and accumulation process evolved along the time. The received SINR of each active and non-deferred IoT device (i.e., IoT devices with packets and do not deferred by the ACB or the back-off mechanism) in each time slot is captured, and compared with the SINR threshold $\gamma_{\rm th}$ to determine the success or failure of each RACH attempt. Furthermore, in the ACB scheme, we also simulate that each IoT device generates a random number $q∈ [0, 1]$ and compares with the ACB factor $\rm P_{ACB}$ to determine whether the current RACH is deferred, and in the back-off scheme, we capture all RACH failures and practically defer RACH attempts of these IoT devices for the next $t_{\rm BO}$ time slots.
%Importantly, the real buffer at each IoT device is simulated to capture the packets arrival and accumulation process evolved along the time. Furthermore, to determine whether a packet is successfully transmitted or accumulated in the buffer, we also capture the received SINR of each active and non-deferred IoT device in each time slot and compare with the SINR threshold $\gamma_{\rm th}$. 
%We set the SINR detection threshold and the SINR transmission success threshold as equal, and are represented by SINR threshold $\gamma_{\rm th}$. 
In all figures of this section, we use “Ana.” and “Sim.” to abbreviate “Analytical” and “Simulation”, respectively. Unless otherwise stated, we set the same new packets arrival rate for each time slot ($\varepsilon^1_{\rm New}=\varepsilon^2_{\rm New}=\cdots=\varepsilon^m_{\rm New}=0.1$ packets/ms), $\rho = -90$ dBm, ${\sigma ^2 = −90 }$ dBm, $\lambda_{B} = 10$ BS/km$^2$, $\lambda_{Dp} = 100$ IoT deivces/preamble/km$^2$, $\alpha = 4$, and $\gamma_{\rm th}=-10$ dB. In the back-off scheme, we set that failure transmission IoT device waits 1 time slot before retransmission in the back-off scheme.

Fig. \ref{fig:5} plots the preamble transmission success probability $\cal P$ versus the density ratio $\lambda_{Dp}/\lambda_B$ for various path-loss exponents ($\alpha$) and various time duration ($\tau_c+\tau_g$), where the analytical plots of the preamble transmission success probability in a single time slot ${\cal P}$ is calculated using (\ref{Pt 1}) (${\cal R} = 1$). We first see the well match between the analysis and the simulation results, which validates the accuracy of developed single time slot mathematical framework. We observe that increasing the density ratio between the IoT devices and the BSs decreases the preamble transmission success probability of the $1$st time slot, due to the increasing aggregate interference from more IoT devices transmitting signals simultaneously. We also notice that increasing the interval duration between RACHs decreases the preamble transmission success probability. This can be explained by the reason that the number of new arrival packets during longer interval duration increases, and leads to higher non-empty probability of IoT devices as shown in (\ref{THINNING P 1}).

\vspace*{-0.6cm}
\begin{figure}[htbp!]
    \begin{center}
    \begin{minipage}[t]{0.46\textwidth}
    \centering
        \includegraphics[width=1\textwidth]{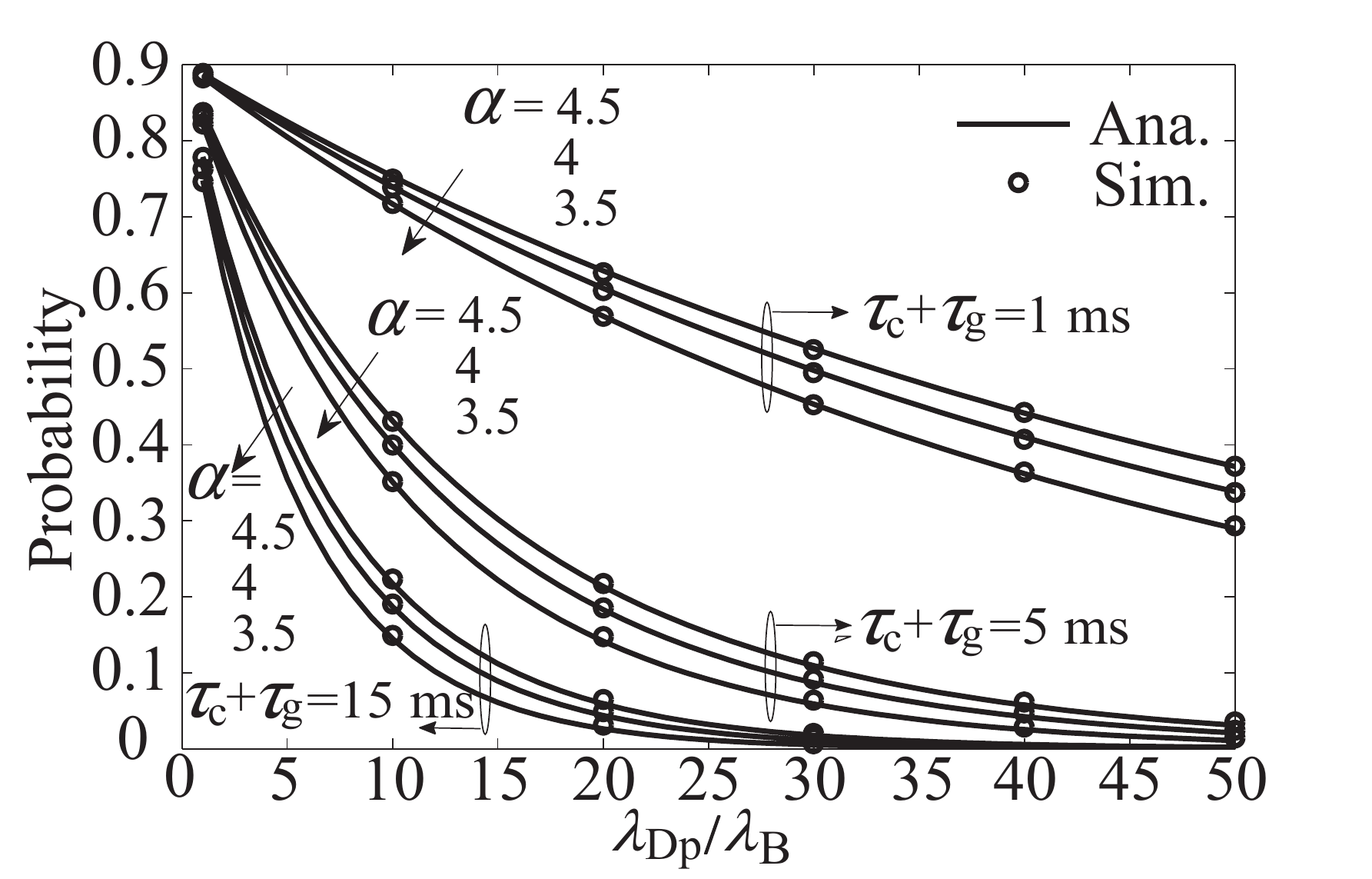}
        \vspace*{-1.2cm}
        \caption{\scriptsize Preamble transmission success probability.}
            \label{fig:5}
    \end{minipage}
    \hspace*{+1cm}
        \begin{minipage}[t]{0.45\textwidth}
    \centering
        \includegraphics[width=1\textwidth]{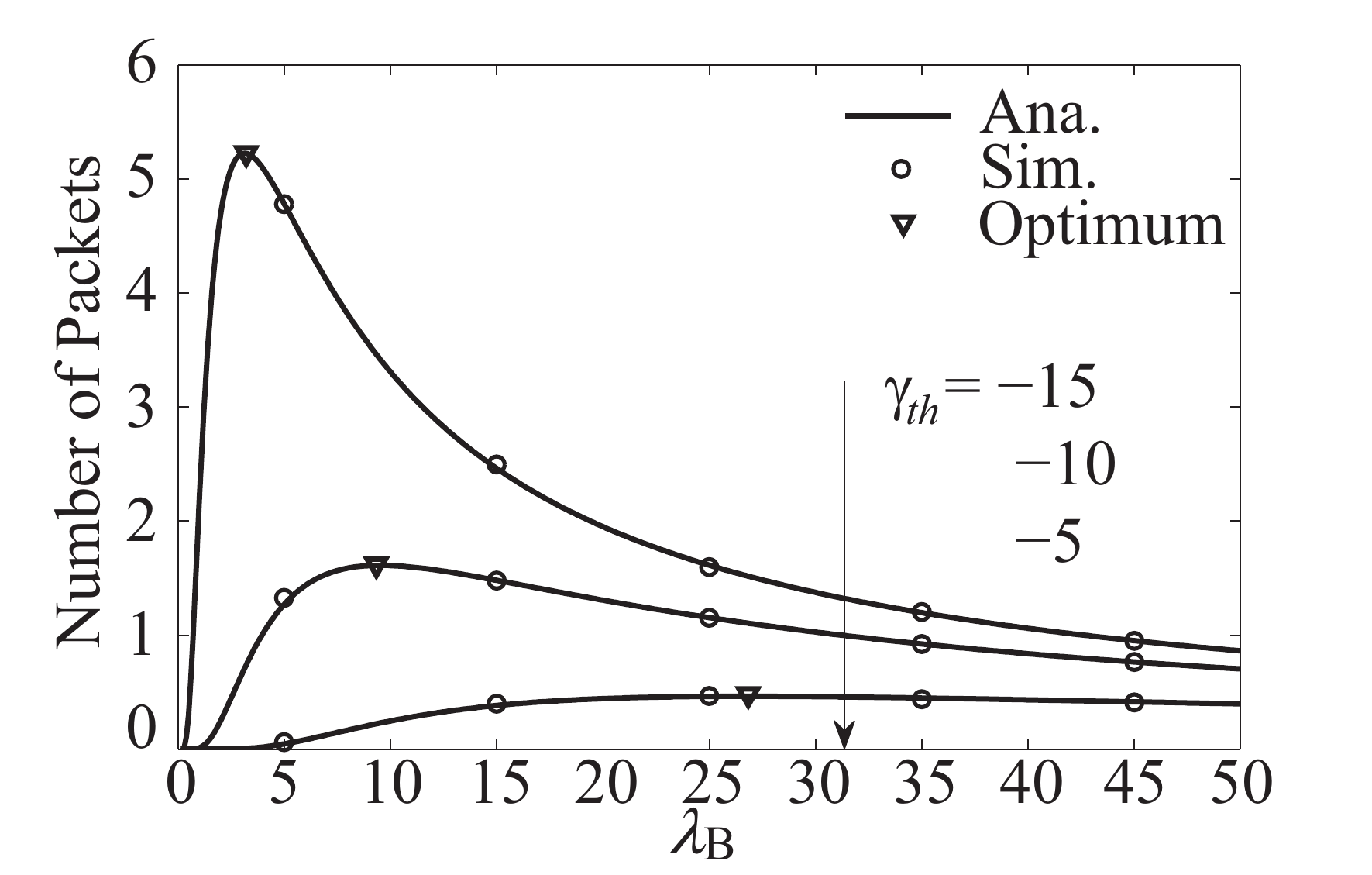}
        \vspace*{-1.2cm}
        \caption{\scriptsize The number of received packets per BS.}
                \label{fig:6}
        \end{minipage}
    \end{center}
\end{figure}
\vspace*{-0.9cm}

Fig. \ref{fig:6} plots the number of received packets per BS ${\cal C}$ in a single time slot versus the density of BSs $\lambda_B$ for various SINR threshold $\gamma_{\rm th}$ (${\cal R} = 1$). We set $\lambda_{Dp} = 500$ IoT deivces/preamble/km$^2$. The analytical curves for the number of received packets per BS are plotted using (\ref{Ct 2}), and the optimal BSs densities that achieve the maximum number of received packets per BS are plotted using (\ref{Ct max 1}). We can see that the calculated optimal BS densities well predict the optimal density points achieving the maximum number of received packets per BS. The first increasing trend of the number of received packets per BS is mainly due to the improvement of the average received SINR, whereas the decreasing trend after $\lambda^*_{B}$ is mainly due to the decreased average number of associated IoT devices of each BS leading to the reduction in channel resources utilization.

Fig. \ref{fig:7} plots the preamble transmission success probabilities of a random IoT device in each time slot with the baseline scheme, the ACB scheme, and the back-off scheme using (\ref{Ct max 1}). For each scheme, the preamble transmission success probabilities decrease with increasing time, due to that the intensity of interfering IoT devices grows with increasing non-empty probability of each IoT device, caused by the increasing average number of accumulated packets. For each scheme, its preamble transmission success probability with $\gamma_{\rm th}=-5$ dB decreases faster than that with $\gamma_{\rm th}=-10$ dB, due to the higher chance of the accumulated packets being reduced for $\gamma_{\rm th}=-10$ dB leading to relatively lower average non-empty probability of each IoT device. Interestingly, we observe that the preamble transmission success probabilities of a random IoT device in each time slot always follow ACB(${\rm P}_{\rm ACB}= 0.5$)$>$back-off$>$ACB(${\rm P}_{\rm ACB}= 0.9$)$>$baseline scheme (except the $1$st time slot, where the back-off procedure is not executed), this is because more strict congestion control schemes reduce more access requests from the side of IoT devices, which decrease the aggregate interference in the network.

\vspace*{-0.6cm}
\begin{figure}[htbp!]
    \begin{center}
        \begin{minipage}[t]{0.46\textwidth}
    \centering
        \includegraphics[width=1\textwidth]{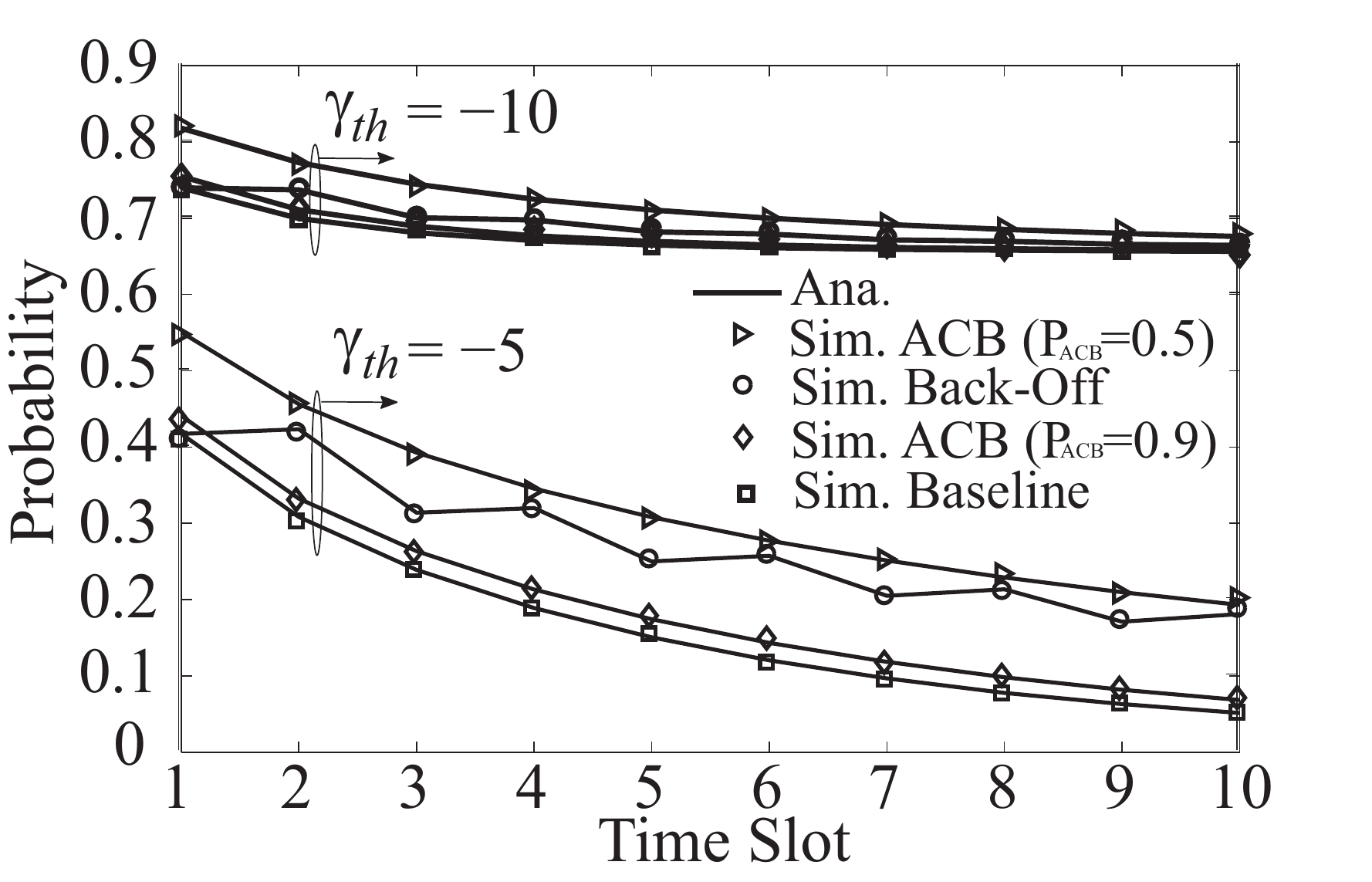}
        \vspace*{-1.2cm}
        \caption{\scriptsize Preamble transmission success probability of each time slot.}
                \label{fig:7}
        \end{minipage}
    \hspace*{+1cm}
            \begin{minipage}[t]{0.46\textwidth}
    \centering
        \includegraphics[width=1\textwidth]{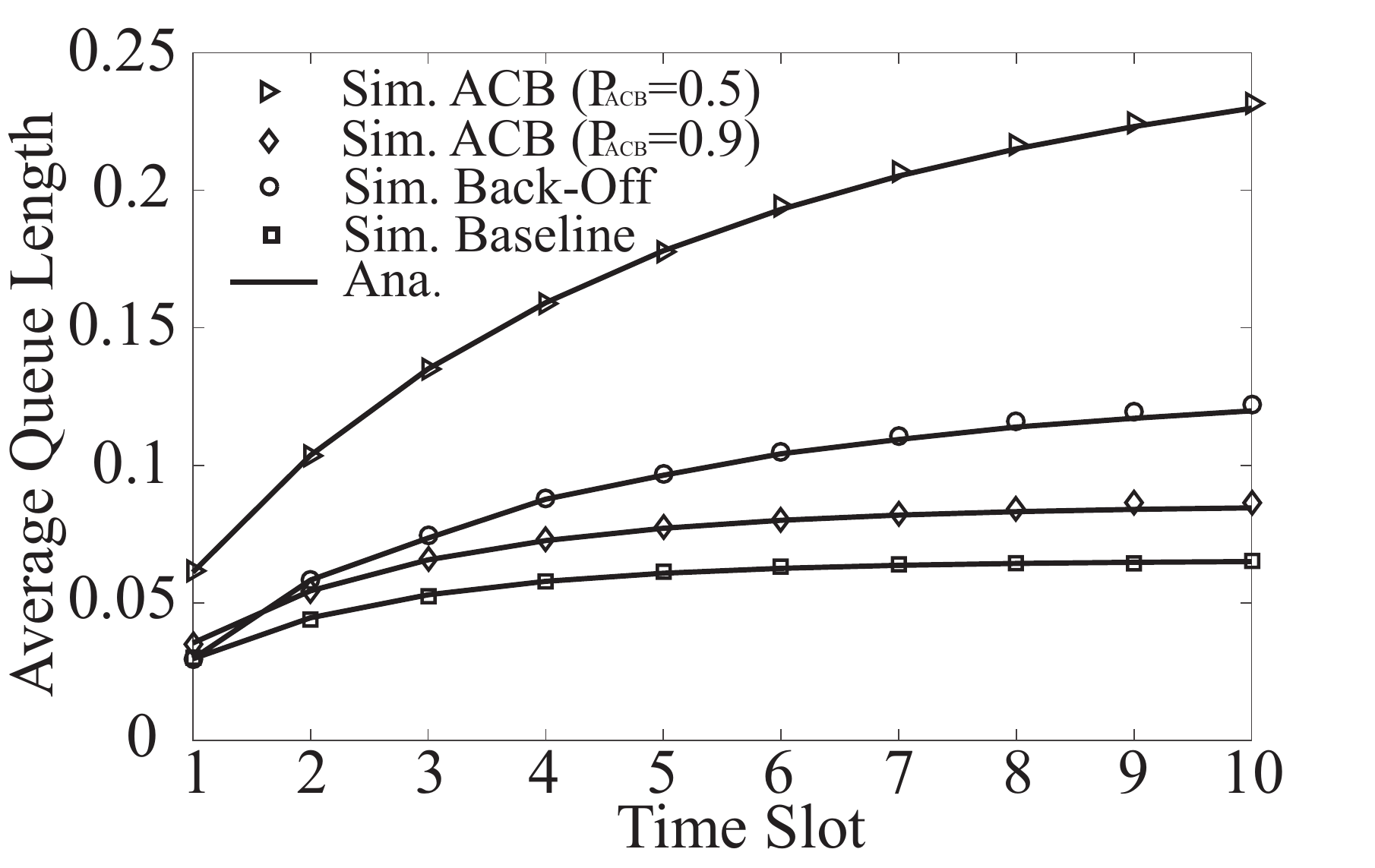}
        \vspace*{-1.2cm}
        \caption{\scriptsize Average Queue Length of each time slot.}
            \label{fig:8}
    \end{minipage}
    \end{center}
\end{figure}
\vspace*{-0.9cm}

We also notice that for $\gamma_{\rm th} =-10$ dB case, the preamble transmission success probabilities with the ${\rm P}_{\rm ACB}= 0.5$ slightly outperform that of ACB scheme (${\rm P}_{\rm ACB}= 0.9$), and the gap between them reduces with increasing time, whilst for $\gamma_{th}=-5$ dB, the preamble transmission success probabilities with ${\rm P}_{\rm ACB}= 0.5$ is much greater than that with ${\rm P}_{\rm ACB}= 0.9$, and such gap increases with increasing time. This is because for $\gamma_{\rm th}=-5$ dB, the ACB scheme ${\rm P}_{\rm ACB}= 0.5$ is more efficient than ${\rm P}_{\rm ACB}= 0.9$, in terms of providing higher average SINR by reducing the probability of queue flushing, but reversely for $\gamma_{\rm th}=-10$ dB, the ACB scheme (${\rm P}_{\rm ACB}= 0.5$) has less access requests leading to lower utilization of channel resources. The preamble transmission success probability of a randomly chosen IoT device with back-off scheme is fluctuated, due to the alternation of high load and low load network condition in each time slot. Furthermore, for $\gamma_{\rm th}=-10$ dB case, the fluctuation become stable quickly, due to the accumulated packets can be handled much quicker.

In Fig \ref{fig:8}, we plots the average queue length with $\gamma_{\rm th}=-10$ dB using (\ref{t 1 m 71}). We observe that the average queue length of the baseline scheme, the ACB(${\rm P}_{\rm ACB}= 0.9)$ scheme and the back-off scheme gradually becomes steady (i.e., they become unchanging in the $10$th time slot). This is due to that these schemes provides relatively faster buffer flushing that can maintain the average accumulated packets in an acceptable level. The average queue lengths follow baseline$<$ACB(${\rm P}_{\rm ACB}= 0.9)<$back-off$<$ACB(${\rm P}_{\rm ACB}= 0.5$) scheme, which shed lights on the buffer 
flushing capability of each scheme in this network condition.

%In the scenario of $\tau_g=1$ ms, the preamble transmission success probabilities of all the congestion control schemes become close to the baseline scheme with the increase of time slot, due to the baseline scheme provide quicker buffer flushing in such network condition (relative low load), but other congestion control schemes more or less block the packets, which leads to inefficient usage of the RACH channel resources and reduce the probability of flushing the queue, such that increases the intensity of interfering IoT devices, and the average SINR decrease much quicker than the baseline scheme. In the scenario of $\tau_g=5$ ms, each IoT devices has less opportunities to request to access than the scenario of $\tau_g=1$ ms, such that the network is relatively overloaded. It can be seen that the gap between the preamble transmission success probabilities of all congestion control schemes, and the baseline scheme is non-decreasing, due to the buffer flushing of the baseline scheme is same or less than other congestion control schemes in this network condition.

\vspace*{-0.6cm}
\begin{figure}[htbp!]
    \begin{center}
    \subfigure[\scriptsize $\gamma_{th} = -8$ dB]{\includegraphics[width=0.46\textwidth]{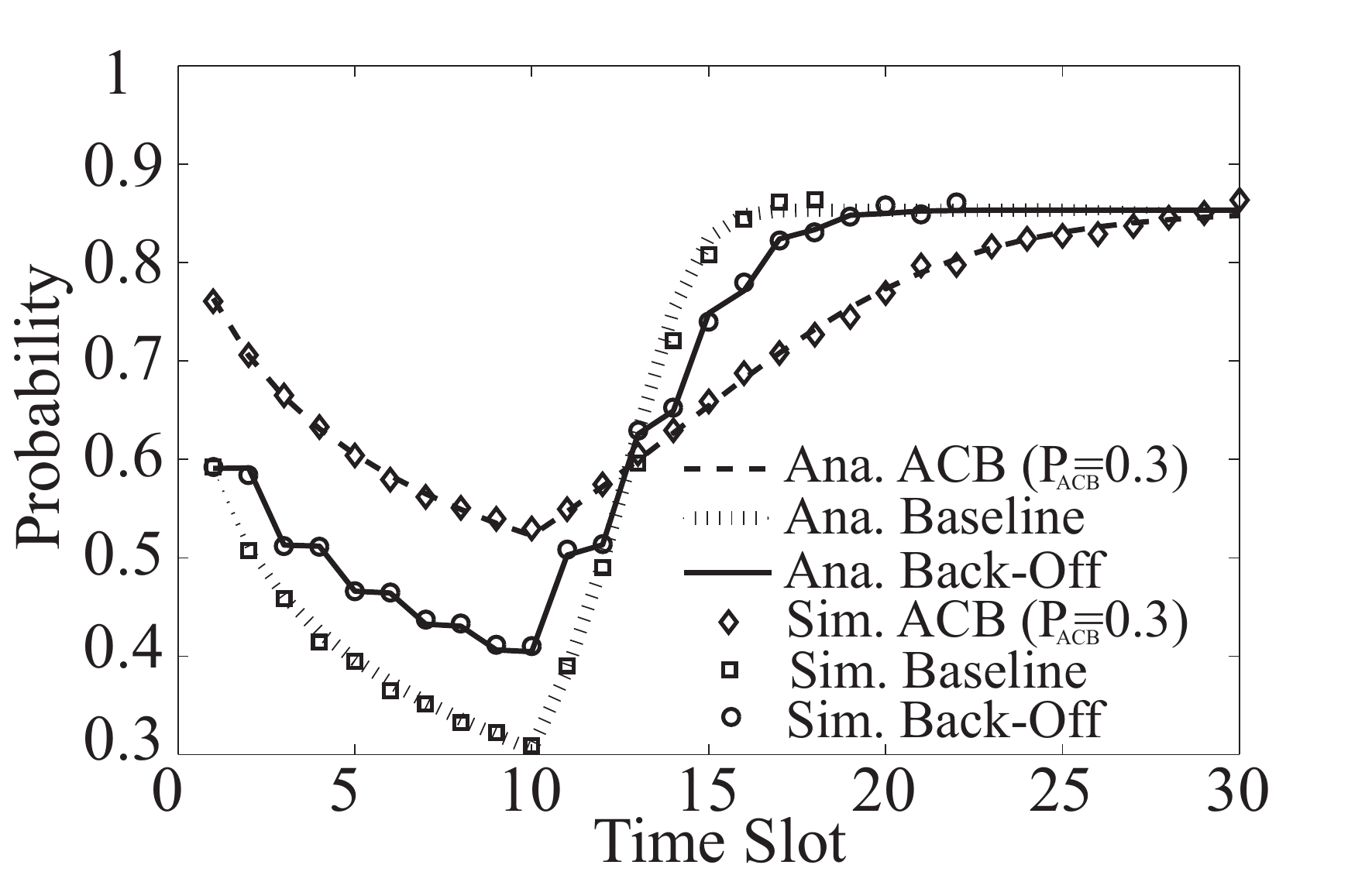}}
    \subfigure[\scriptsize $\gamma_{th} = -6$ dB]{\includegraphics[width=0.42\textwidth]{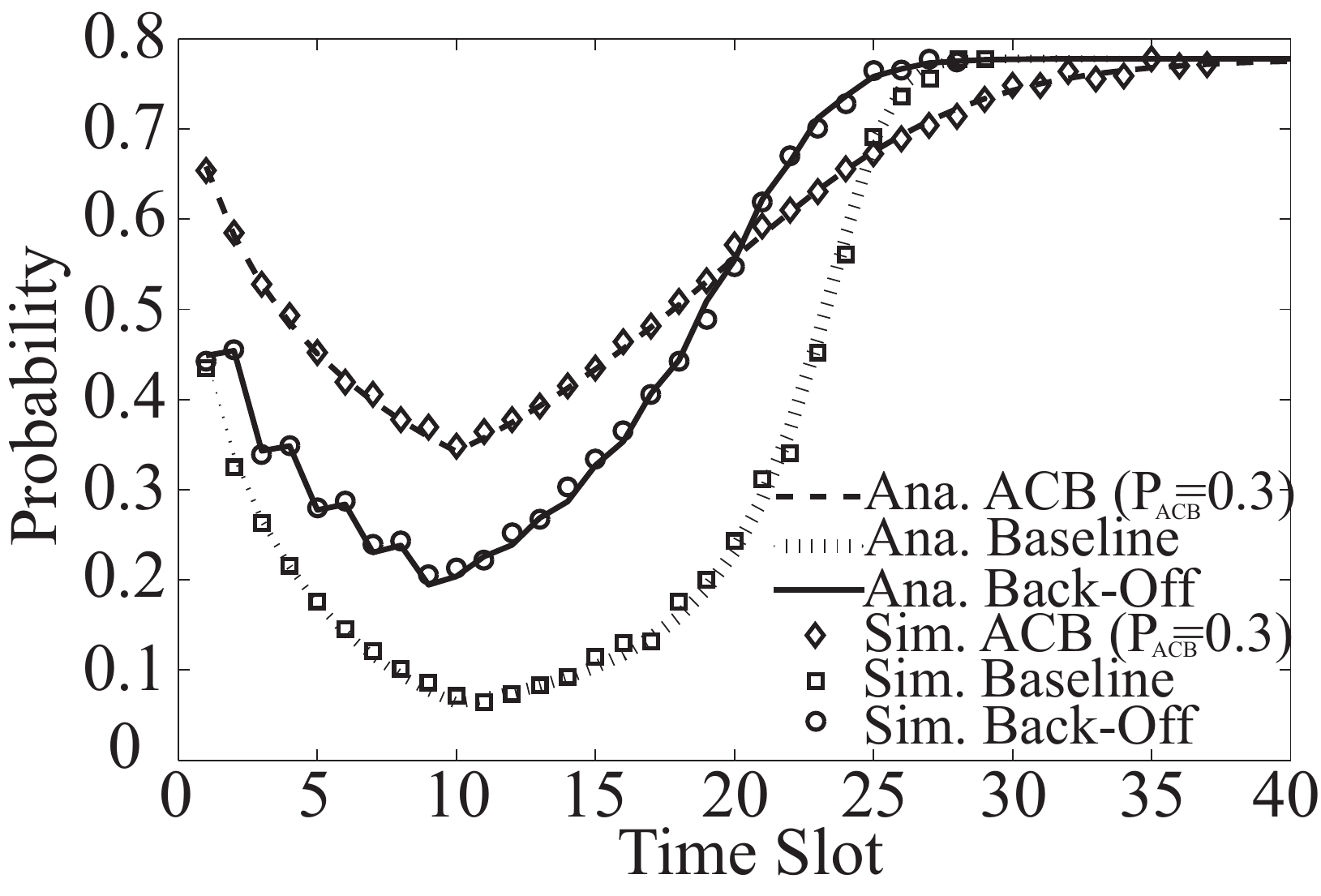}}
    \vspace*{-0.5cm}
        \caption{\scriptsize Preamble transmission success probability of each time slot.}
        \label{fig:9}
    \end{center}
\end{figure}
\vspace*{-0.9cm}

Fig. \ref{fig:9} plots the preamble transmission success probability of a random IoT device in each time slot with the baseline scheme, the ACB scheme, and the back-off scheme. We set $\tau_c+\tau_g=5$ ms, ACB factor ${\rm P}_{\rm ACB}= 0.3$, and new arrival traffics only happen in the first 10 time slots ($\varepsilon^m_{\rm New}=0$ for $m>10$). Note that this simulation method with new arrival traffics happen in first several time slots is to examine how well the network can handle bursty traffic, where similar practical simulations has been tested in \cite{3GPP1046632010RACH,3GPP1037422010RACH}. In both Fig. \ref{fig:9}(a) and Fig. \ref{fig:9}(b), the preamble transmission success probabilities decrease in the first 10 time slots, due to increasing traffic (new packets arrived) leading to increasing active probabilities of IoT devices. After first 10 time slots, these probabilities increase with time, due to decreasing traffic (i.e., no new packets arrive) leading to decreased active probabilities of IoT devices. After most of the accumulated packets are delivered with time, the preamble transmission success probabilities reaches the stable ceiling. Interestingly, we see that the preamble transmission success probabilities in Fig. \ref{fig:9}(a) ($\gamma_{\rm th}=-8$ dB) become stable earlier than that in Fig. \ref{fig:9}(b) ($\gamma_{\rm th}=-6$ dB), due to that the higher chance of the accumulated packets being reduced in lower threshold case.

The preamble transmission success probability of the baseline scheme increases rapidly after first $10$ time slots and outperforms other two schemes after first $12$th time slots in Fig. \ref{fig:9}(a), but it increases relatively slowly after first $10$ time slots and only outperforms that of the ACB scheme after first $25$ time slots in Fig. \ref{fig:9}(b), due to that the baseline scheme provide faster buffer flushing, which leads to lower chance of the accumulated packets being reduced in relatively higher loaded network condition due to the high aggregate interference. The back-off scheme performs better than the baseline scheme in the first $10$ time slots (except $1$st time slot where back-off is not executed), due to that it automatically defers the retransmission requests and control the congestion in the overloaded network condition. Interestingly, it gradually outperforms the ACB scheme with strictly ACB factor ${\rm P}_{\rm ACB}= 0.3$ after the first $10$ time slots, due to that the back-off scheme automatically release the blocking of packets and provide faster buffer flushing than the ACB scheme in the non-overloaded network condition.

%We notice that the ACB scheme with ${P}_{ACB}= 0.3$ provides highest preamble transmission success probabilities in the first few time slots, but it is then surpassed by other schemes after the $12$th time slot in Fig. 7 (a) and after the $25$th time slots in Fig.7 (b), respectively. This is due to that the ACB scheme become inefficient in the low traffic scenario (after $10$th time slot), where the probabilities of removing packet from the queue gradually become lower than other schemes with decreasing the traffic load.

\vspace*{-0.6cm}
\begin{figure}[htbp!]
    \begin{center}
        \subfigure[\scriptsize The mean of preamble transmission success probabilities]{\includegraphics[width=0.46\textwidth]{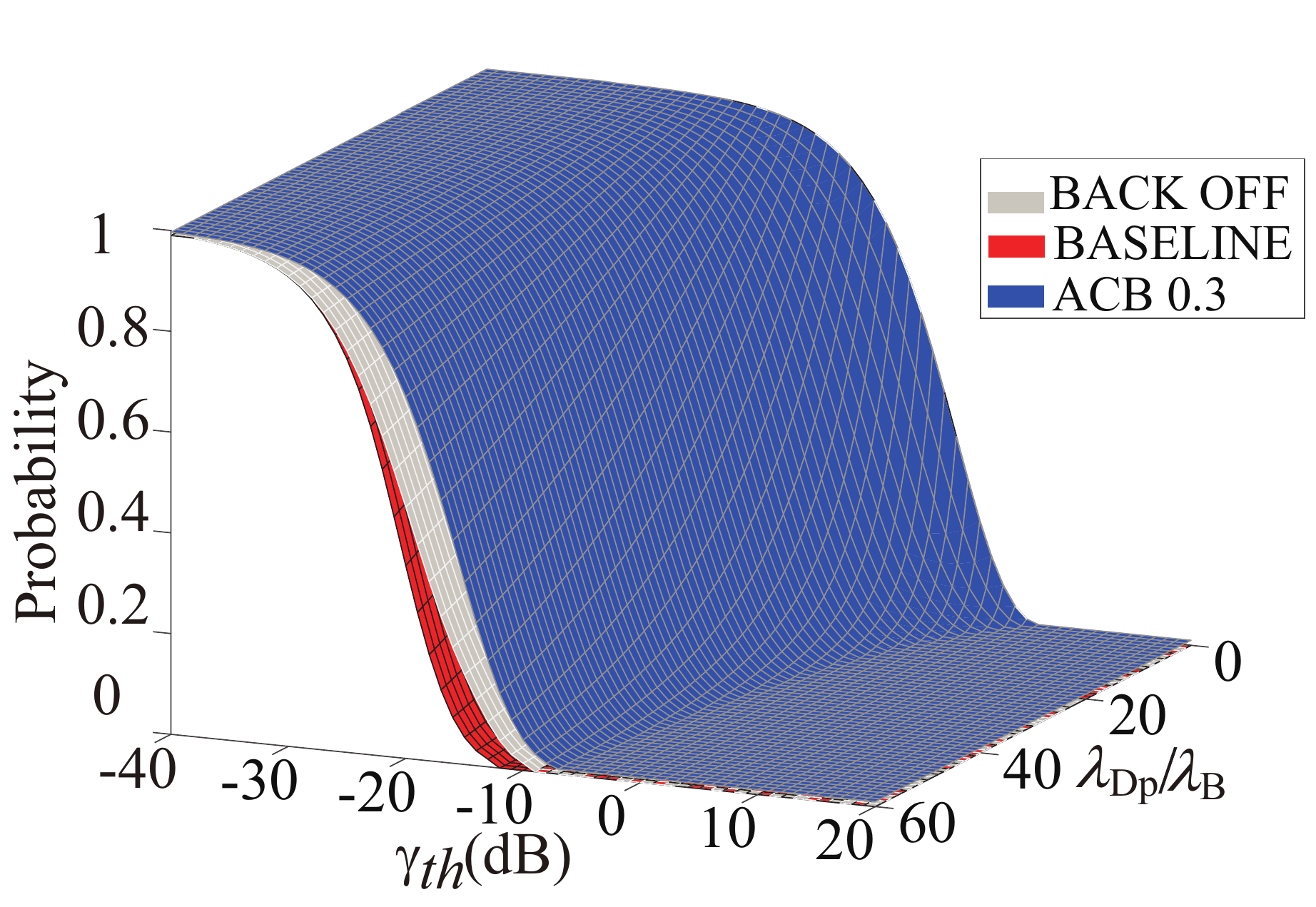}}
        \hspace*{+1cm}
         \subfigure[\scriptsize The mean of transmission capacities per BS per preamble]{\includegraphics[width=0.46\textwidth]{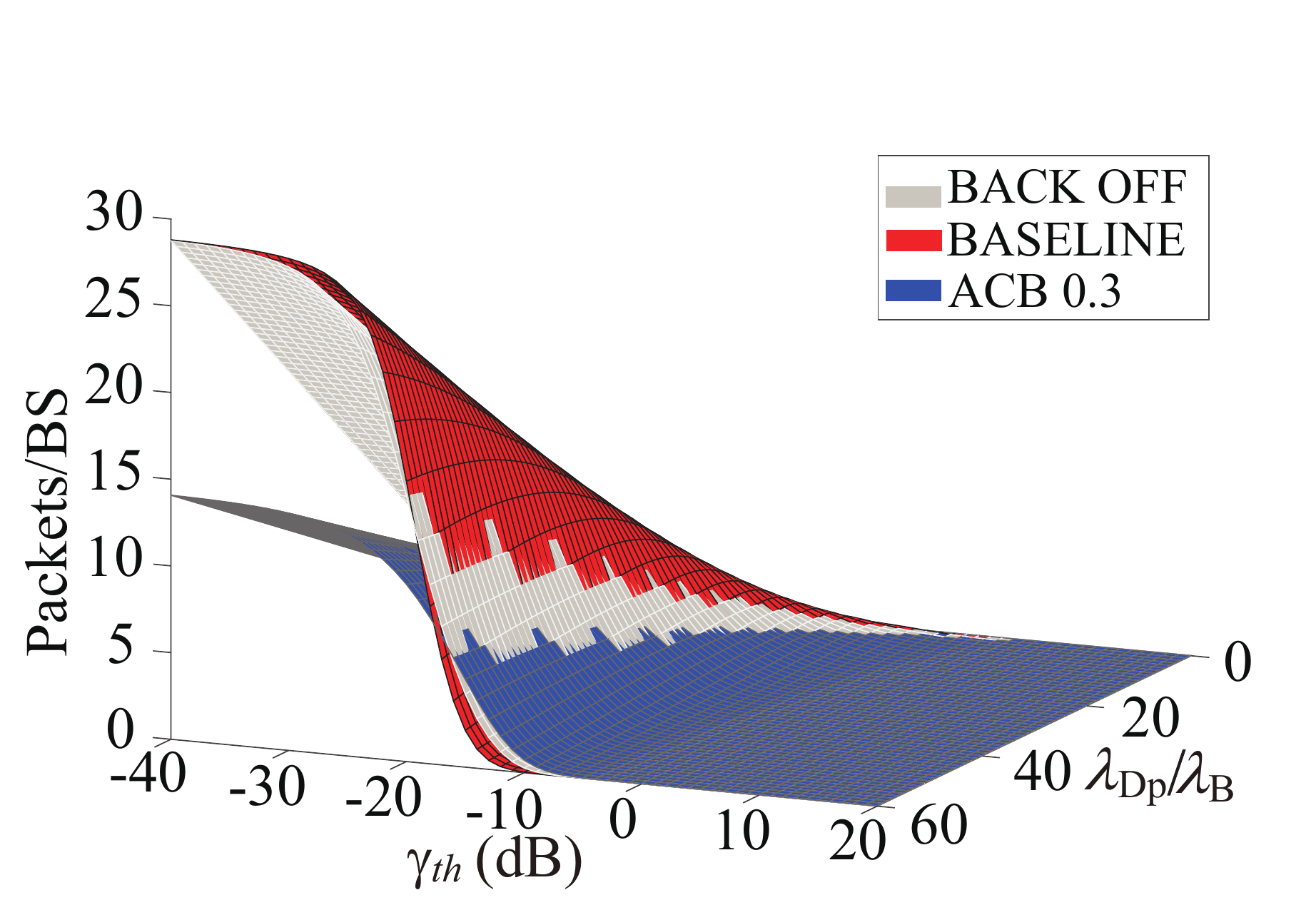}}
         \vspace*{-0.4cm}
        \caption{\scriptsize The mean of preamble transmission success probabilities and the transmission capacities per BS per preamble}
        \label{fig:10}
    \end{center}
\end{figure}
\vspace*{-0.9cm}

In Fig. \ref{fig:10}(a) and Fig. \ref{fig:10}(b), we plot the mean of preamble transmission success probabilities and the mean of numbers of received packets per BS over 10 time slots with each scheme, respectively. We set $\tau_g=1$ ms and ACB factor ${\rm P}_{\rm ACB}= 0.3$. Note that the new traffics arrival happen in every time slot. In Fig. \ref{fig:10}(a), the ACB scheme always outperforms the other two schemes, and the mean of probabilities of the back-off scheme is slightly higher than that of the baseline scheme before $\gamma_{\rm th} = -25$ dB, and then such gap between the back-off scheme and the baseline scheme increase with increasing $\gamma_{\rm th}$, which is due to that the back-off scheme blocks more packets.

In Fig. \ref{fig:10}(b) we observe that 1) For $-40 \le \gamma_{\rm th} \le -25$ dB, the mean of numbers of received packets per BS with the back-off scheme is slightly lower than the baseline scheme, but nearly double that of the ACB scheme, due to the preamble transmission success probability is close to 1 as shown in Fig. \ref{fig:10}(a), and thus less packets are blocked in the IoT device in the back-off scheme. 2) For $-25 < \gamma_{\rm th} \le -15$ dB, the mean of numbers of received packets per BS with the baseline and the back-off schemes decrease dramatically and reduce to same level with the ACB scheme. The back-off scheme gradually outperforms the baseline scheme around the $\gamma_{\rm th} = -20$ and $-15$ dB, because the back-off scheme gradually blocks more IoT devices, and provides better network condition as well as higher probabilities of removing packets from the queue. 3) For $-15 < \gamma_{\rm th} \le -5$ dB, the ACB scheme outperforms the other schemes, which showcases that the ACB scheme with a relatively strict ACB factor can provide improved successful transmission in overloaded network.

%In the fig. 7(a), traffic only arrived during the first 5 time slots ($\varepsilon^m_{New}=0$ when $m>5$ms), and in the fig. 7(b), traffic only arrived during the first 10 time slots ($\varepsilon^m_{New}=0$ when $m>10$ms). The network of fig. 7(b) tolerates more traffic than fig. 7(a), such that the curves becomes stable (most of the packets are successfully transmitted) later than fig. 7(a). In both fig. 7(a) and fig. 7(b), the ACB scheme performs best in the traffic arriving duration, due to that it block more access requests than the other two schemes. The back-off scheme performs worse but close to the ACB scheme (${\cal P}_{ACB}= 0.5$) in the traffic arriving duration, due to that the congestion handling capacity of the back-off scheme just slightly lower than the ACB scheme in the overloaded network condition, and then it gradually outperforms the ACB scheme, due to that it automatically release the blocking of packets and provide faster buffer flushing than the ACB scheme in the non-overloaded network condition. The baseline scheme provide fastest buffer flushing in non-overloaded network condition, such that its preamble transmission success probability outperforms the other two schemes in the $17$th time slot in fig. 7(a), however, the network of fig. 7(b) is much more overloaded than fig. 7(a), such that its preamble transmission success probability still worse than the other two schemes, in spite of it increase dramatically during $30$th to $40$th time slots.

\vspace*{-0.5cm}
\section{Conclusion}
\vspace*{-0.2cm}
In this paper, we developed a spatio-temporal mathematical model to analyze the RACH of cellular-based mIoT networks. We first analyzed RACH in the single time slot, and provide the preamble detection probability performed on a randomly chosen BS, preamble transmission success probability performed on a BS associated with a randomly chosen IoT device. We then derived the preamble transmission success probabilities of a randomly chosen IoT device with baseline, ACB, and back-off schemes by modeling the queue evolution over different time slot. Our numerical results show that the ACB and back-off schemes outperform the baseline scheme in terms of the preamble transmission success probability. We also show that the baseline scheme outperforms the ACB and back-off schemes in terms of the number of received packets per BS for light traffic, and the back-off scheme performs closing to the optimal performing scheme in both light and heavy traffic conditions.

\appendices
\numberwithin{equation}{section}
\vspace*{-0.3cm}
\section{ A Proof of Lemma \ref{lemma1}}\label{INTER IF}
\vspace*{-0.2cm}
The Laplace Transform of aggregate inter-cell interference can be derived as
\vspace*{-0.1cm}
\begin{align} \label{INTER I 3}
{{\cal L}_{{I_{{\mathop{ int}} er}}}}(s) & 
 \mathop  = \limits^{(a)} {E_{\widehat{\cal Z}_{out} }} \Big[ \mathop \prod \limits_{{u_i} \in \widehat{\cal Z}_{out} } {E_{{P_i}}}{E_{{h_i}}} \left[ {e^{ - s{P_i}{h_i}{{\left\| {{u_i}} \right\|}^{ - \alpha }}}} \Big] \right]
 \nonumber \\
&\mathop  = \limits^{(b)} \exp \Big( { - 2\pi {{{{{\cal T}^m}{{\cal R}^m}\lambda_{Dp} }}}\int_{{{(P/\rho )}^{\frac{1}{\alpha }}}}^\infty  {{E_P}{E_h}\left[1 - {e^{ - sPh{x^{ - \alpha }}}}\right]xdx} } \Big)
\nonumber  \\
& \mathop  = \limits^{(c)} \exp \left( { - 2\pi {{{{{\cal T}^m}{{\cal R}^m}\lambda_{Dp} }}}\int_{{{(P/\rho )}^{^{\frac{1}{\alpha }}}}}^\infty  {{E_P}\left[1 - \frac{1}{{1 + sP{x^{ - \alpha }}}}\right]xdx} } \right)
\nonumber \\
& \mathop  = \limits^{(d)} \exp \Big( { - 2\pi {{{{{\cal T}^m}{{\cal R}^m}\lambda_{Dp} }}}{s^{\frac{2}{\alpha }}}{E_P}[{P^{\frac{2}{\alpha }}}]\int_{{{(s\rho )}^{\frac{{ - 1}}{\alpha }}}}^\infty  {\frac{y}{{1 + {y^\alpha }}}dy} } \Big) ,
\end{align}
where  $s = \gamma_{\rm th} /\rho $,  $E_x[{*}]$ is the expectation with respect to the random variable $x$, (a) follows from independence between ${\lambda_{Dp} }$, $P_i$, and $h_i$, (b) follows from the probability generation functional (PGFL) of the PPP, (c) follows from the Laplace Transform of $h$, and (d) obtained by changing the variables $y = \frac{x}{{{{(SP)}^{\frac{1}{\eta }}}}}$. The $k$th moments of the transmit power is expressed as \cite{HElSawy2014Stochastic}

\vspace*{-0.3cm}
\begin{align}\label{POWER}
{E_P}[{P^k }] = \frac{{\rho _{}^k \gamma (\frac{{k \alpha }}{2} + 1,\pi \lambda_B {{(\frac{{{P}}}{\rho })}^{\frac{2}{\alpha }}})}}{{{{(\pi \lambda_B )}^{\frac{{k \alpha }}{2}}}(1 - {e^{ - \pi \lambda_B {{(\frac{{{P}}}{\rho })}^{\frac{2}{\alpha }}}}})}} ,
\end{align}
where $\gamma (a,b) = \int_0^b {{t^{a - 1}}{e^{ - t}}dt}$ is the lower incomplete gamma function. As mentioned earlier, the transmit power of IoT device is large enough for uplink path-loss inversion, while not violating its own maximum transmit power constraint, and thus 
The moments of the transmit power is obtained as
\vspace*{-0.3cm}
\begin{align}\label{POWER MAX}
{E_P}[P^{\frac{2}{\alpha }} ] = \frac{{\rho _{}^{\frac{2}{\alpha }}}}{{\pi \lambda_B }} .
\end{align}
Substituting \eqref{POWER MAX} into \eqref{INTER I 3}, we derive
the Laplace Transform of aggregate inter-cell interference.

\section{ A Proof of Lemma \ref{lemma2}}\label{Intra IF}

The Laplace Transform of aggregate intra-cell interference is conditioned on known the number of interfering intra-cell IoT devices $Z_B$ given as

\vspace*{-0.3cm}
\begin{align}\label{Intra IF 1}
 {{\cal L}_{{{\cal I}_{{\mathop{ int}} ra}}}}(s) & =
 \sum\limits_{{\rm{n = 0}}}^\infty  {{\mathbb P} \left\{ {{Z_B} = n} \right\}\left( {\left. {E\left[ {{e^{ - sI}}} \right]} \right|{Z_B = n}} \right)}
\nonumber \\
&= {\mathbb P} \left\{ {{Z_B} = 0} \right\} + \sum\limits_{{\rm{n = 1}}}^\infty  {{\mathbb P} \left\{ {{Z_B} = n} \right\} \left({E_{{h_n}}}\left[ {\exp \left( { - s\sum\limits_1^n {\rho {h_n}} } \right)} \right] \right)}
\nonumber \\
&  \mathop  = \limits^{(a)}  {\mathbb P} \left\{ {{Z_B} = 0} \right\}{\rm{ + }}\sum\limits_{{\rm{n = 1}}}^\infty  {{\mathbb P} \left\{ {{Z_B} = n} \right\}{\left( {\frac{{1 }}{{1 + s\rho }}} \right)^n}} ,
\end{align}
where $s = \gamma_{\rm th} /\rho $, ${\mathbb P} \left\{ {{Z_B} = n} \right\}$ is the probability of the number of interfering intra-cell IoT devices ${Z_B}=n$ given in (\ref{BS chosen 4}), and $(a)$ follows from the Laplace Transform of $h_n$. After some mathematical manipulations, we proved (\ref{Intra I 2}) in Lemma \ref{lemma2}.

\bibliographystyle{IEEEtran}

\bibliography{IEEEabrv,RA_bib}

\end{document}